\documentclass[submission,copyright,creativecommons]{eptcs}
\providecommand{\event}{AiML 2026} 

\usepackage{aiml26}
\usepackage{iftex}
\usepackage{amsmath,mathrsfs}
\usepackage{amssymb}
\usepackage{amscd}
\usepackage{proof}
\usepackage{color}
\usepackage{enumerate}
\usepackage{multicol}
\usepackage{url}
\usepackage[all]{xy}
\usepackage[colorinlistoftodos,prependcaption,textsize=tiny]{todonotes}
\usepackage[shortlabels,inline]{enumitem}
\usepackage[strings]{underscore}
\usepackage[T1]{fontenc}

\newcommand{\mb}[1]{\mathbf{#1}}
\newcommand{\ms}[1]{\mathsf{#1}}
\newcommand{\mt}[1]{\mathtt{#1}}
\newcommand{\Ge}[1]{\mathsf{G}(#1)}
\newcommand{\Hi}[1]{\mathsf{H}(#1)}
\newcommand{\Sub}{\mathsf{Sub}}
\newcommand{\inset}[2]{\left\{\, {#1} \,|\, {#2} \,\right\}}
\newcommand{\setof}[1]{\left\{ {#1} \right\}}

\newcommand{\A}{\mathsf{A}}
\newcommand{\D}{\mathcal{D}}

\newcommand{\Prop}[1]{\mathsf{Prop}(#1)}
\newcommand{\Agt}[1]{\mathsf{Agt}(#1)}
\newcommand{\X}[1]{\mathsf{X}(#1)}

\newcommand{\Part}[2]{\langle (#1);(#2) \rangle}

\newtheorem{claim2}{\textbf{Claim}}

\title{Analytic Cut in Epistemic Logics with Distributed Knowledge}
\author{Ryo Murai
\institute{Independent Researcher}
\email{ryo.murai11@gmail.com}
\and
Sizhuo Liu
\institute{Graduate School of Humanities and Human Sciences\\
Hokkaido University\\
Hokkaido, Japan}
\email{liuliusizhuo@outlook.com}
\and
Katsuhiko Sano
\institute{Faculty of Humanities and Human Sciences\\
Hokkaido University\\
Hokkaido, Japan}
\email{v-sano@let.hokudai.ac.jp}
}

\newcommand{\titlerunning}{Analytic Cut in Epistemic Logics with Distributed Knowledge}
\newcommand{\authorrunning}{R. Murai \& S. Liu \& K. Sano}

\hypersetup{
  bookmarksnumbered,
  pdftitle    = {\titlerunning},
  pdfauthor   = {\authorrunning},
  pdfsubject  = {EPTCS},               
  pdfkeywords = {Sequent Calculus, Distributed Knowledge} 
}

\begin{document}
\maketitle

\begin{abstract}
Distributed knowledge is a notion of group knowledge studied in multi-agent epistemic logic. Semantically, the distributed knowledge of a group is interpreted via an accessibility relation given by the intersection of the epistemic accessibility relations of the agents in that group. This paper investigates sequent calculi for epistemic logics of distributed knowledge based on $\mathbf{K45}$, $\mathbf{KD45}$, and $\mathbf{S5}$. While cut elimination holds in existing sequent calculi for modal logics $\mathbf{K45}$ and $\mathbf{KD45}$, it fails in all the systems mentioned above. Instead, we establish the analytic cut property for all three systems by adapting Takano’s (2018) strategy, which restricts the cut formulas to the set of subformulas of the conclusion of the cut rule. As a corollary, the Craig interpolation theorem holds for all logics considered. We also show that all proof-theoretic results remain valid when the empty group is allowed for the distributed-knowledge operator, in which case the distributed knowledge for the empty group is interpreted as the global modality.
\end{abstract}

\section{Introduction}
Distributed knowledge is a notion of group knowledge in multi-agent epistemic logic~\cite{Fagin_2003_Reasoning} whose interpretation has varied in the literature. In this paper we follow the interpretation introduced in~\cite{Agotnes_2017_resolving}, i.e., a group $G$ is said to have distributed knowledge of $\varphi$ if an {\em external} observer with access to the epistemic states of all members of $G$ can know that $\varphi$. In Kripke semantics, while individual knowledge (notation: $\Box_{a}\varphi$, read as ``agent $a$ knows that $\varphi$") is interpreted via an accessibility relation $R_{a}$ for each agent $a$, distributed knowledge for a group $G$ is interpreted via the intersection $\bigcap_{a \in G} R_a$ of the epistemic accessibility relations of the agents in $G$. Accordingly, $D_G \varphi$ holds at a state $w$ if and only if $\varphi$ holds at all states $v$ such that $(w,v) \in \bigcap_{a \in G} R_a$. Reflexivity, transitivity, Euclideanness, and symmetry are preserved when passing from the individual level to the group level, i.e., under intersection of epistemic accessibility relations, whereas seriality is {\em  not} preserved (see~\cite[Lemma 3.2]{Agotnes_2021_group}).

{\em Sequent calculus}, introduced by Gentzen in the 1930s, is a proof-theoretic framework where derivations are formulated as sequents $\Gamma \Rightarrow \Delta$ (read as ``if all formulas in $\Gamma$ hold, then some formula in $\Delta$ holds"), where $\Gamma$ and $\Delta$ are finite multisets of formulas~\cite{Gentzen_1964_investigations,Gentzen_1965_investigations}. 
A crucial aspect of sequent calculus is the cut rule
\[
\infer[(Cut)]{\Gamma, \Pi \Rightarrow \Delta, \Sigma}{
\Gamma \Rightarrow \Delta, \varphi
&
\varphi, \Pi \Rightarrow \Sigma
},
\]
and the associated cut elimination theorem (Gentzen’s {\em Hauptsatz}), which states that every derivable sequent admits a cut-free derivation. As emphasized in~\cite{troelstra_2000_basic}, cut elimination is closely connected to proof normalization and the structural analysis of derivations.

For modal logics, Gentzen-style sequent calculi $\mathsf{G}(\mb{K45})$ and $\mathsf{G}(\mb{KD45})$ for $\mb{K45}$ and $\mb{KD45}$, respectively, were shown to be cut-free in~\cite[Theorem~2]{Shvarts_1989_gentzen}. However, as shown in~\cite[p.~116]{Ohnishi_1959_gentzen}, cut elimination fails for the sequent calculus $\mathsf{G}(\mb{S5})$ corresponding to modal logic $\mb{S5}$. As an alternative, Takano~\cite{Takano_1992_subformula} showed that every application of the cut rule in $\mathsf{G}(\mb{S5})$ can be transformed into an {\em analytic} cut, i.e., a cut rule in which the cut formula is a subformula of the conclusion. Moreover, Takano later provided a semantic analysis of sequent calculi for modal logics in~\cite{Takano_2018_semantical}, where cut-free sequent calculi such as $\mathsf{G}(\mb{K45})$ and $\mathsf{G}(\mb{KD45})$ were semantically shown to enjoy cut elimination, while sequent calculi that lack cut elimination such as $\mathsf{G}(\mb{S5})$ were shown to satisfy the extended subformula property. These results are established by proving the semantic completeness of sequent calculi where the cut rule is removed or restricted to an analytic cut. A key ingredient of the semantic proof is the notion of $\Xi$-{\em partial valuation}, i.e., a subformula-closed finite set of formulas, which originates from~\cite{Takano_2001_modified}, that was later applied in~\cite{Maruyama_2003_temporal,Maruyama_2003_combined},~ \cite{Takano_2018_semantical},~\cite{Ono_2022_analytic} and~\cite{Udatsu_2025_craig} to establish cut elimination or the analytic cut property for the sequent calculi for temporal epistemic logics, modal logics, bi-intuitionistic tense logic and awareness logic, respectively.

Model-theoretic investigations of distributed knowledge have been extensively documented in, for example,~\cite{Gerbrandy_1999_bisimulations,Roelofsen_2007_distributed,Agotnes_2017_resolving,Agotnes_2021_group}, whereas proof-theoretic studies, namely sequent calculi for logics with distributed knowledge, remain comparatively limited. A labelled $G3$-style sequent calculus, i.e., a calculus without structural rules in which each formula is equipped with a label, was introduced in~\cite{Hakli_2008_proofa}. A Gentzen-style sequent calculus based on $\mb{S4}$ was proposed in~\cite{Pliuskevicius_2008_termination}, in which the distributed knowledge operator is {\em not} parameterized by a group $G$. Gentzen-style and Kanger-style sequent calculi for $\mb{S5}$ were later studied in~\cite{Giedra_2010_cuta} in the same fashion. Label-free Gentzen-style sequent calculi for multi-agent epistemic propositional logics $\mb{K}_{D}$, $\mb{KD}_{D}$, $\mb{K4}_{D}$, $\mb{S4}_{D}$, and $\mb{S5}_{D}$ with group-parameterized distributed knowledge operators were introduced in~\cite{Murai_2020_craig}. Cut elimination was established for these systems except for $\Ge{\mb{S5}_{D}}$, which is the sequent calculus for $\mb{S5}_{D}$. Moreover, the Craig interpolation theorem was obtained for the four systems via a method based on~\cite{Maehara_1961_craig}. Furthermore, although groups are usually defined as {\em non-empty}, the operator $D_{\emptyset}$ corresponds naturally to the global modality when the empty set is allowed as a group. 

A sequent calculus $\Ge{\mb{S5}_{D}}$ for $\mb{S5}_{D}$ was introduced in~\cite{Murai_2020_craig,Subformula_2023_Murai}. This paper extends that line of research by introducing sequent calculi $\Ge{\mb{K45}_{D}}$ and $\Ge{\mb{KD45}_{D}}$ for $\mb{K45}_{D}$ and $\mb{KD45}_{D}$, respectively, while also revisiting $\Ge{\mb{S5}_{D}}$. Surprisingly, although cut elimination holds for $\mathsf{G}(\mb{K45})$ and $\mathsf{G}(\mb{KD45})$, it {\em fails} not only for $\Ge{\mb{S5}_{D}}$, but {\em also} for $\Ge{\mb{K45}_{D}}$ and $\Ge{\mb{KD45}_{D}}$. Therefore, by adapting Takano’s semantic approach~\cite{Takano_2018_semantical}, we establish the analytic cut property for all three systems. As a consequence, the Craig interpolation theorem is obtained via the method in~\cite{Maehara_1961_craig}. We note that conditions on {\em both} propositional variables and {\em agents} are taken into account, which is a relatively new result for logics for distributed knowledge. To the best of the authors' knowledge, such conditions have recently been considered only in~\cite{Murai_2020_craig,Umemura_2025_craig} and~\cite{Bilkova_2026_agent}. Although the Craig interpolation theorem in this paper has been established in~\cite{Umemura_2025_craig}, the proof is model-theoretic rather than proof-theoretic as in~\cite{Maehara_1961_craig}. In fact, a proof-theoretic result was presented earlier by the first and third authors in~\cite{Subformula_2023_Murai}. As emphasized in~\cite{Murai_2020_craig}, this result suggests that the logics with distributed knowledge are ``good” expansions of basic modal logic, since the Craig interpolation theorem {\em does not} hold for some expansions of basic modal logic~\cite{TenCate_2005_interpolation}. Moreover, we show that when the empty group is admitted, the operator $D_{\emptyset}$ corresponds naturally to the global modality, and that the main proof-theoretic results are preserved under this extension.

The paper is organized as follows. Section~\ref{sec:preliminary} introduces the syntax and semantics of distributed epistemic logics based on $\mb{K45}$, $\mb{KD45}$, and $\mb{S5}$ together with their Hilbert systems and pseudo-model semantics. Section~\ref{sec:sequent calculi} presents the sequent calculi $\mathsf{G}(\mb{K45}_{D})$, $\mathsf{G}(\mb{KD45}_{D})$, and $\mathsf{G}(\mb{S5}_{D})$. Section~\ref{sec:a-cut D} first demonstrates that cut elimination fails for these systems, then establishes the analytic cut property, while the Craig interpolation theorem is proved in Section~\ref{sec:CIT}. In Section~\ref{sec:global modality}, we extend the framework by introducing the global modality corresponding to distributed knowledge of the empty group and show that the analytic cut property and interpolation results are preserved under this extension. Finally, Section~\ref{sec:conclusion} concludes the paper. 

\section{Preliminaries}
\label{sec:preliminary}
Let $\mathsf{Prop}$ be a countably infinite set of propositional variables and $\mathsf{Ag}$ a non-empty finite set of agents. A {\em group} is a {\em non-empty} set $G \subseteq \mathsf{Ag}$ o f agents and the set of all groups is denoted by 
\[\mathsf{Grp}:=\inset{G\in\mathcal{P}(\mathsf{Ag})}{G\ne\emptyset}.
\] 
The syntax $L$ for distributed epistemic logic is defined inductively as follows:
\[
L \ni \varphi :: =  p\mid\bot\mid\neg\varphi \mid \varphi \to \psi \mid D_{G}\varphi,
\]
\noindent where $p \in \mathsf{Prop}$ and $G \in\mathsf{Grp}$. Boolean connectives $\land,\lor,\leftrightarrow$ and the truth constant $\top$ are defined in the usual manner. We also define the operator $\Box_{a}\varphi$ for individual knowledge or belief as $\Box_{a}\varphi:= D_{\{a\}}\varphi$. For a finite set $\Delta$, $\bigwedge \Delta$ and $\bigvee \Delta$ are the conjunction or disjunction of all formulas in $\Delta$ (when $\Delta$ is empty, $\bigwedge \Delta$ := $\top$ and $\bigvee \Delta$ := $\bot$). As noted above, an expression of the form $D_{\emptyset}\varphi$ is not a well-formed formula, since we have excluded $\emptyset$ from our definition of groups.

We move on to the semantics. A {\em frame} $F$ is a tuple $(W, (R_{a})_{a \in \mathsf{Ag}})$ where $W$ is a non-empty set of states and $R_{a} \subseteq W \times W$ is a binary relation on $W$ for each $a \in \mathsf{Ag}$. A {\em model} $M$ is a tuple $(W, (R_{a})_{a \in \mathsf{Ag}}, V)$ where $(W, (R_{a})_{a \in \mathsf{Ag}})$ is a frame and $V:\mathsf{Prop}\to\mathcal{P}(W)$ is a valuation function. Given a model $M = (W, (R_{a})_{a \in \mathsf{Ag}}, V)$ and a state $w \in W$, the notion of {\em $\varphi$ being true} at $w$ in $M$ (notation: $M, w \models \varphi$) is defined inductively as follows:
\begin{alignat*}{3}
    &M, w \models p  &\quad& \text{iff} &\quad &w \in V(p),\\
    &M, w \models \bot  &\quad&&\quad& \text{Never},\\
    &M, w \models \neg \varphi  &&\text{iff}&&M, w \not\models \varphi,\\
    &M, w \models \varphi \to \psi&&\text{iff}&&M, w \not\models \varphi\ \text{or}\ M,w\models \psi,\\
    &M, w \models D_{G}\varphi&&\text{iff}&&\text{for all}\ v \in W, (w,v)\in{\bigcap}_{a\in G}R_{a} \text{ implies } M, v \models \varphi.
\end{alignat*}

\noindent A pair $(M,w)$ of a model $M$ and a state $w$ in $M$ is called a {\em pointed model}. A formula $\varphi$ is {\em valid} in a model $M$ (notation: $M \models \varphi$) if $M,w \models \varphi$ for all $w \in W$. A formula $\varphi$ is {\em valid} in a frame $F$  (notation: $F \models \varphi$)  if $(F, V) \models \varphi$ for all valuations $V$. Given a class $\mathbb{F}$ of frames, $\varphi$ is {\em valid} in $\mathbb{F}$ if $F \models \varphi$ for every frame $F \in \mathbb{F}$. 

To capture various aspects of individual knowledge operators, we consider a class of frames $\mathbb{F}$ in which all accessibility relations $R_{a}$ for $a \in \mathsf{Ag}$ satisfy the same set of frame properties, defined as follows.
\begin{definition}
\label{dfn:frames-D}
We define $\mathbb{F}_{\mathbf{K45}}$ to be the class of all frames such that each $R_{a}$ is transitive and Euclidean, $\mathbb{F}_{\mathbf{KD45}}$ to be the class of all frames such that each $R_{a}$ is transitive, Euclidean and serial ($\forall w\in W\exists v(wR_{a}v)$), and $\mathbb{F}_{\mathbf{S5}}$ to be the class of all frames such that each $R_{a}$ is reflexive, transitive, and Euclidean (equivalently, symmetric). 
We define the distributed epistemic logics $\mb{K45}_{D}$, $\mb{KD45}_{D}$ and $\mb{S5}_{D}$ to be the sets of all $L$-formulas that are valid in the class of $\mathbb{F}_{\mathbf{K45}}$, $\mathbb{F}_{\mathbf{KD45}}$ and $\mathbb{F}_{\mathbf{S5}}$, respectively.
\end{definition}

\noindent  As proved in~\cite[Lemma 3.2]{Agotnes_2021_group}, reflexivity, transitivity, Euclideanness and symmetry are preserved under taking the intersection ${\bigcap}_{a\in G}R_{a}$ of binary relations\footnote{However, seriality is {\em  not} preserved. As shown in~\cite[Lemma 3.2]{Agotnes_2021_group}, we can construct a $\mb{KD45}$ model in which $R_{a}$ and $R_{b}$ are serial while $R_{a}\cap R_{b}$ is not.
Thus, as discussed in~\cite{Agotnes_2021_group}, the axiom $\neg D_{G} \bot$ in~\cite{Fagin_1995_reasoning} is not valid, hence the original axiomatization of $\mb{KD45}_{D}$ in~\cite{Fagin_1995_reasoning} is unsound. This axiomatization was subsequently corrected in~\cite{Fagin_2003_Reasoning}, where completeness is reported, although without a detailed proof.
}. It is well-known that all logics in Definition \ref{dfn:frames-D} can be axiomatized (see~\cite{Fagin_2003_Reasoning}).

\begin{table}[htbp]
    \centering
    \caption{Hilbert system $\Hi{\mb{K}_{D}}$, $\Hi{\mb{K45}_{D}}$, $\Hi{\mb{KD45}_{D}}$ and $\Hi{\mb{S5}_{D}}$}
    \label{table:hilbert-systems-D}
    \begin{tabular}{|c|c|}
        \hline
        \multicolumn{2}{|c|}{All the Axioms and Rules of $\Hi{\mb{K}_{D}}$}\\
        \hline
        ($\mathtt{Taut}$) &  all instances of propositional tautologies \\
        ($\mathtt{Incl}$) &  $D_{G}\varphi\to D_{H}\varphi$ ($G\subseteq H$) \\
        ($\mathtt{K}$) & $D_{G}(\varphi \to \psi) \to (D_{G} \varphi \to D_{G} \psi)$ \\
        ($\mathtt{MP}$) & From $\varphi \to \psi$ and $\varphi$ infer $\psi$\\
        ($\mathtt{Nec}$) & From $\varphi$ infer $D_{G}\varphi$\\
        \hline
    \end{tabular}
    \quad
    \begin{tabular}{|c|c|}
        \hline 
        \multicolumn{2}{|c|}
        {Additional Axiom Schemes}\\
        \hline
        ($\mathtt{T}$) & $D_{G}\varphi \to \varphi$ \\
        ($\mathtt{D}$) & $\neg D_{\{a\}}\bot$ \\
        ($\mathtt{4}$) & $D_{G}\varphi \to D_{G}D_{G}\varphi$\\
        ($\mathtt{5}$) & $\neg D_{G}\varphi \to D_{G}\neg D_{G}\varphi$\\
        \hline
    \end{tabular}
\end{table}
\begin{definition}
The Hilbert system $\Hi{\mb{K}_{D}}$ is defined as in Table~\ref{table:hilbert-systems-D}. 
The Hilbert system $\Hi{\mb{K45}_{D}}$ is defined as an axiomatic expansion of $\Hi{\mb{K}_{D}}$ with $(\mathtt{4})$ and $(\mathtt{5})$ of Table~\ref{table:hilbert-systems-D}. Moreover, the Hilbert systems $\Hi{\mb{KD45}_{D}}$ and $\Hi{\mb{S5}_{D}}$ are defined as axiomatic expansions of $\Hi{\mb{K45}_{D}}$ with $(\mathtt{D})$ and $(\mathtt{T})$ of Table~\ref{table:hilbert-systems-D}, respectively.    
\end{definition}

The following notion of {\em pseudo-model} is used in, for instance,~\cite{Fagin_2003_Reasoning,Gerbrandy_1999_bisimulations}, as an intermediate step for semantic completeness proofs of epistemic logics with distributed knowledge. However, in this paper, we use pseudo-models to establish semantically the analytic property of sequent calculi in Section~\ref{sec:a-cut D}.

\begin{definition}
    A {\em pseudo frame} $F$ is a tuple $(W, (R_{G})_{G\in\mathsf{Grp}})$ where $W$ is a set of states, $R_{G} \subseteq W \times W$ is a binary relation on $W$ such that for each $G,H\in \mathsf{Grp}$ and $G\subseteq H$ implies $R_{H}\subseteq R_{G}$. A {\em pseudo model} $M$ is a tuple $(W, (R_{G})_{G\in\mathsf{Grp}},V)$ where $(W, (R_{G})_{G\in\mathsf{Grp}})$ is a pseudo frame and $V$ is a function from $\mathsf{Prop}$ to $\mathcal{P}(W)$. 
\end{definition}      
    Given a frame $F=(W,(R_{a})_{a\in\ms{Ag}})$, we define $R_{G}:=\bigcap_{a\in G}R_{a}$. Then $(W, (R_{G})_{G\in\mathsf{Grp}})$ is a pseudo frame. 
    The notion of {\em $\varphi$ being true} at $w$ in $M$ is defined in the same way as at the beginning of Section~\ref{sec:preliminary} except that
\[
M,w\models D_{G}\varphi \text{ iff }\text{ for all}\ v \in W, w R_{G} v \text{ implies } M, v \models \varphi.
\]
\begin{definition}
\label{dfn:pseudo frames-D}
We define $\mathbb{P}_{\mathbf{K45}_{D}}$ to be the class of all pseudo frames such that each $R_{G}$ is transitive and Euclidean, $\mathbb{P}_{\mathbf{KD45}_{D}}$ to be the class of all pseudo frames $F$ = $(W, (R_{G})_{G\in\mathsf{Grp}})$ such that $F \in \mathbb{P}_{\mathbf{K45}_{D}}$ and each $R_{\{a\}}$ is serial, $\mathbb{P}_{\mathbf{S5}_{D}}$ to be the class of all pseudo frames such that each $R_{G}$ is reflexive, transitive, and Euclidean (equivalently, symmetric).
\end{definition}

\begin{fact}
    Let $\Lambda\in\{\mb{K45}_{D},\mb{KD45}_{D},\mb{S5}_{D}\}$. The following are all equivalent:
    \begin{enumerate*}[label=\arabic*., ref=\arabic*]
        \item $\varphi$ is a theorem of $\Hi{\Lambda}$;
        \item $\varphi$ is valid in $\mathbb{P}_{\Lambda}$;
        \item $\varphi$ is valid in $\mathbb{F}_{\Lambda}$.
    \end{enumerate*}
\end{fact}
\begin{proof}
First, item 1 implies item 2 by the soundness of $\Hi{\Lambda}$, which can be established straightforwardly. Moreover, item 2 implies item 3 since $\mathbb{F}_{\Lambda}$ can be regarded as a subclass of $\mathbb{P}_{\Lambda}$. So it remains to show that item 3 implies item 1, which is a result already established in the literature. A proof sketch for $\Hi{\mb{K}_{D}}$ and $\Hi{\mb{K45}_{D}}$ can be found in~\cite[Theorem 3.4.1]{Fagin_2003_Reasoning} and~\cite[Theorem 3.30]{Gerbrandy_1999_bisimulations}, respectively. Detailed proofs are provided in~\cite[Theorem 4.7]{Agotnes_2021_group} for $\Hi{\mb{KD45}_{D}}$ and in~\cite[Theorem 10]{Wang_2013_publica} for $\Hi{\mb{S5}_{D}}$. 
\end{proof}

\section{Sequent Calculi for Distributed Knowledge Logics}
\label{sec:sequent calculi}
In this section we introduce sequent calculi $\Ge{\mb{K45}_{D}}$, $\Ge{\mb{KD45}_{D}}$ and $\Ge{\mb{S5}_{D}}$ built on $\mathbf{LK_{0}}$ (see~\cite{Takano_1992_subformula}), which is the propositional fragment of a sequent calculus $\mathbf{LK}$~\cite{Gentzen_1964_investigations,Gentzen_1965_investigations} of the first-order classical logic. A {\em sequent} $\Gamma \Rightarrow \Delta$ is a pair of finite multisets $\Gamma$ and $\Delta$ of formulas read as ``if all formulas in $\Gamma$ hold, then at least one formula in $\Delta$ holds". Moreover, we use $D_{\{a\}}\Gamma$ to denote $\inset{D_{\{a\}}\varphi}{\varphi\in\Gamma}$ and $\Gamma,\Delta$ to denote the multiset union $\Gamma\cup\Delta$.

\begin{definition}
\label{dfn:seq-calc-ml}
A sequent calculus $\mathbf{LK_{0}}$ consists of the following.
\begin{itemize}[itemsep=-1pt, parsep=-2pt, topsep=3pt]
\item Axioms:
\[
\infer[(\mathtt{id})]{\varphi \Rightarrow  \varphi}{}
\quad 
\infer[(\bot)]{\bot\Rightarrow}{}
\]
\item Structural Rules: 
\[ 
\infer[(\Rightarrow w)]{\Gamma \Rightarrow \Delta, \varphi}{\Gamma \Rightarrow \Delta}
\quad
\infer[(w \Rightarrow)]{\varphi, \Gamma \Rightarrow \Delta}{\Gamma \Rightarrow \Delta}
\quad
\infer[(\Rightarrow c)]{\Gamma \Rightarrow \Delta, \varphi}{\Gamma \Rightarrow \Delta, \varphi, \varphi}
\quad
\infer[(c \Rightarrow)]{\varphi, \Gamma \Rightarrow \Delta}{\varphi, \varphi, \Gamma \Rightarrow \Delta}
\]
\item Logical Rules:
\[
\infer[(\Rightarrow  \neg)]{\Gamma \Rightarrow \Delta, \neg \varphi}{\varphi, \Gamma \Rightarrow \Delta}
\quad
\infer[(\neg \Rightarrow  )]{\neg \varphi, \Gamma \Rightarrow \Delta}{\Gamma \Rightarrow \Delta, \varphi}
\]
\[
\infer[(\Rightarrow  \to)]{\Gamma \Rightarrow \Delta, \varphi \to \psi}{\varphi, \Gamma \Rightarrow \Delta, \psi}
\quad
\infer[(\to \Rightarrow)]{\varphi \to \psi, \Gamma, \Pi \Rightarrow \Delta, \Sigma}{\Gamma \Rightarrow \Delta, \varphi & \psi, \Pi \Rightarrow \Sigma}
\]
\item Cut:
\[
\infer[(Cut)_\varphi]{\Gamma, \Pi \Rightarrow \Delta, \Sigma}{\Gamma \Rightarrow \Delta, \varphi & \varphi, \Pi \Rightarrow \Sigma}
\]
\end{itemize}
The sequent calculus $\Ge{\mb{K45}_{D}}$ is $\mathbf{LK_{0}}$ expanded with the following rule:
\[
\infer[(\Rightarrow D^{\mb{K45}_{D}})]{D_{G_{1}}\varphi_{1},\dots,D_{G_{m}}\varphi_{m}\Rightarrow D_{H_{1}}\psi_{1},\dots,D_{H_{n}}\psi_{n}, D_{G}\chi}{\varphi_{1},\dots,\varphi_{m},D_{G_{1}}\varphi_{1},\dots,D_{G_{m}}\varphi_{m}\Rightarrow D_{H_{1}}\psi_{1},\dots,D_{H_{n}}\psi_{n},\chi},
\]
where $G_{i},H_{j}\subseteq G$ for any $1\le i\le m$ and $1\le j\le n$. The sequent calculus $\Ge{\mb{KD45}_{D}}$ expands $\Ge{\mb{K45}_{D}}$ with the following:
\[
\infer[(D^{\mb{KD45}_D})]{D_{\{a\}}\Gamma \Rightarrow D_{\{a\}}\Delta}{\Gamma, D_{\{a\}}\Gamma\Rightarrow D_{\{a\}}\Delta}.
\]
Finally, the sequent calculus $\Ge{\mb{S5}_{D}}$ expands $\mathbf{LK_{0}}$ with the following rules:
\[
\infer[(D \Rightarrow)]{D_{G} \varphi,\Gamma \Rightarrow\Delta}{\varphi,\Gamma\Rightarrow\Delta
}
\quad
\infer[(\Rightarrow D^{\mb{S5}_D})]{D_{G_1}\varphi_1, \dots, D_{G_m}\varphi_m\Rightarrow D_{H_1}\psi_1, \dots, D_{H_n}\psi_n, D_G \chi}{D_{G_1}\varphi_1, \dots, D_{G_m}\varphi_m\Rightarrow D_{H_1}\psi_1, \dots, D_{H_n}\psi_n,\chi},
\]
where $G_{i},H_{j}\subseteq G$ for any $1\le i\le m$ and $1\le j\le n$ in $(\Rightarrow D^{\mb{S5}_D})$.
\end{definition}

\noindent Note that $m$ and $n$ are possibly $0$ in $(\Rightarrow D^{\mb{K45}_{D}})$ and $(\Rightarrow D^{\mb{S5}_D})$. For each calculus, we define the notion of derivability of a sequent as a finite tree generated from axioms $(\mathtt{id})$ and $(\bot)$ by inference rules specific to the calculus. In addition, we use a double line $=$ as an abbreviation of finite applications of structural rules. Given a derivation $\mathcal{D}$, a sequent in the root node of $\D$ is called the {\em end-sequent} and we use $\mathtt{root}(\mathcal{D}) \equiv \Gamma \Rightarrow \Delta$ to mean that the end-sequent of $\mathcal{D}$ is $\Gamma \Rightarrow \Delta$. Moreover we define the {\em height} of $\mathcal{D}$ to be the maximum length of branches in $\mathcal{D}$ from the end-sequent to an axiom.

\begin{proposition}
\label{prop:soundness-D}
Let $\Lambda \in \{\mathbf{K45}_{D}, \mathbf{KD45}_{D}, \mathbf{S5}_{D}\}$. If $\Gamma \Rightarrow \Delta$ is derivable in $\mathsf{G}(\Lambda)$ then $\bigwedge \Gamma \to \bigvee \Delta$ is valid in $\mathbb{F}_{\Lambda}$.
\end{proposition}
\begin{proof}
    By induction on the height of the derivation $\D$ to obtain $\Gamma \Rightarrow \Delta$. For base step, our goal is immediate. For inductive step, we divide the argument depending on the last rule applied. We only comment on the case of $(\Rightarrow D^{\mb{K45}_{D}})$ for $\Ge{\mathbf{K45}_{D}}$:
    \[
\infer[(\Rightarrow D^{\mb{K45}_{D}})]{D_{G_{1}}\varphi_{1},\dots,D_{G_{m}}\varphi_{m}\Rightarrow D_{H_{1}}\psi_{1},\dots,D_{H_{n}}\psi_{n}, D_{G}\chi}{\varphi_{1},\dots,\varphi_{m},D_{G_{1}}\varphi_{1},\dots,D_{G_{m}}\varphi_{m}\Rightarrow D_{H_{1}}\psi_{1},\dots,D_{H_{n}}\psi_{n},\chi},
\]
where $\bigcup^{m}_{i=1}G_{i}\cup\bigcup^{n}_{j=1}H_{j}\subseteq G$. Our goal is to show $\mathbb{F}_{\mb{K45}_{D}}\models\bigwedge_{i=1}^{m}D_{G_{i}}\varphi_{i}\to\bigvee^{n}_{j=1}D_{H_{j}}\psi_{j}\lor D_{G}\chi$. Fix any model $M$ from $\mathbb{F}_{\mb{K45}_{D}}$ and any $w\in W$. Suppose $M,w\models \bigwedge_{i=1}^{m}D_{G_{i}}\varphi_{i}$. We show that $M,w\models\bigvee^{n}_{j=1}D_{H_{j}}\psi_{j}\lor D_{G}\chi$. Suppose $M,w\not\models\bigvee^{n}_{j=1}D_{H_{j}}\psi_{j}$ and we show that $M,w\models D_{G}\chi$. Fix any $v\in W$ such that $(w,v)\in \bigcap_{a\in G}R_{a}$. It suffices to show that $M,v\models\chi$. By induction hypothesis, we obtain $\mathbb{F}_{\mb{K45}_{D}}\models\bigwedge_{i=1}^{m}\varphi_{i}\land\bigwedge_{i=1}^{m}D_{G_{i}}\varphi_{i}\to\bigvee^{n}_{j=1}D_{H_{j}}\psi_{j}\lor\chi$. Since $G_{i}\subseteq G$, $(w,v)\in\bigcap_{a\in G}R_{a}\subseteq \bigcap_{b\in G_{i}}R_{b}$. Hence $M,v\models\bigwedge_{i=1}^{m}\varphi_{i}$ by our assumption $M,w\models \bigwedge_{i=1}^{m}D_{G_{i}}\varphi_{i}$. Moreover by transitivity we obtain $\mathbb{F}_{\mb{K45}_{D}}\models\bigwedge_{i=1}^{m}D_{G_{i}}\varphi_{i}$, since for every $u\in W$ such that $(v,u)\in\bigcap_{b\in G_{i}}R_{b}$, $(w,u)\in\bigcap_{b\in G_{i}}R_{b}$. Therefore $M,v\models\bigvee^{n}_{j=1}D_{H_{j}}\psi_{j}\lor\chi$. For our goal, it suffices to show $M,v\not\models\bigvee^{n}_{j=1}D_{H_{j}}\psi_{j}$. By $M,w\not\models\bigvee^{n}_{j=1}D_{H_{j}}\psi_{j}$, there exists a state $u_{j}\in W$ for every $1\le j\le n$ such that $(w,u_{j})\in\bigcap_{c\in H_{j}}R_{c}$ and $M,u_{j}\not\models\psi_{j}$. Then by Euclideanness, we obtain $(v,u_{j})\in\bigcap_{c\in H_{j}}R_{c}$, which implies our goal $M,v\not\models\bigvee^{n}_{j=1}D_{H_{j}}\psi_{j}$ together with $M,u_{j}\not\models\psi_{j}$. \qedhere    
\end{proof}

\begin{proposition}
\label{prop:derivability-D}
    Let $\Lambda \in \{\mathbf{K45}_{D}, \mathbf{KD45}_{D}, \mathbf{S5}_{D}\}$. If $\varphi$ is a theorem of $\mathsf{H}(\Lambda)$, then $\Rightarrow \varphi$ is derivable in $\mathsf{G}(\Lambda)$.
\end{proposition}

\begin{proof}
We note that $(Cut)$ is necessary to simulate $(\mathtt{MP})$. We only check the derivability of $(\mathtt{Incl})$ in each system. 
Suppose $G \subseteq H$. For $\Lambda \in \{\mathbf{K45}_{D}, \mathbf{KD45}_{D}, \mathbf{S5}_{D}\}$, the following derivations (left for $\mathbf{K45}_{D}, \mathbf{KD45}_{D}$ and right for $\mathbf{S5}_{D}$) witness $(\mathtt{Incl})$, where the side conditions are satisfied by the assumption $G \subseteq H$.
       \[
        \infer[(\Rightarrow\to)]{\Rightarrow D_{G}\varphi\to D_{H}\varphi}{\infer[(\Rightarrow D^{\mb{K45}_{D}})]{D_{G}\varphi\Rightarrow D_{H}\varphi}{\infer[(w\Rightarrow)]{\varphi, D_{G}\varphi\Rightarrow \varphi}{\varphi\Rightarrow\varphi}}}
        \quad
        \infer[(\Rightarrow\to)]{\Rightarrow D_{G}\varphi\to D_{H}\varphi}{\infer[(\Rightarrow D^{\mb{S5}_{D}})]{D_{G}\varphi\Rightarrow D_{H}\varphi}{\infer[(D\Rightarrow)]{D_{G}\varphi\Rightarrow \varphi}{\varphi\Rightarrow\varphi}}}. \qedhere
        \]
\end{proof}

As shown in~\cite{Shvarts_1989_gentzen}, sequent calculi $\Ge{\mb{K45}}$ and $\Ge{\mb{KD45}}$ for modal logic $\mb{K45}$ and $\mb{KD45}$, respectively, are cut-free, while sequent calculus $\Ge{\mb{S5}}$ for modal logic $\mb{S5}$ is not~\cite[p.~116]{Ohnishi_1959_gentzen}. A sequent $\Rightarrow D_{\{a\}}\neg D_{\{a\}}p, D_{\{a,b\}}p$ is derivable in $\Ge{\mb{K45}_{D}}$ and $\Ge{\mb{KD45}_{D}}$ as follows:
\[
\infer[(Cut)_{D_{\{a\}}p}]{\Rightarrow D_{\{a\}}\neg D_{\{a\}}p, D_{\{a,b\}}p}{\infer[(\Rightarrow D^{\mb{K45}_{D}})]{\Rightarrow D_{\{a\}}\neg D_{\{a\}}p, D_{\{a\}}p}{\infer[(\Rightarrow\neg)]{\Rightarrow \neg D_{\{a\}}p, D_{\{a\}}p}{D_{\{a\}}p\Rightarrow D_{\{a\}}p}} & \infer[(\Rightarrow D^{\mb{K45}_{D}})]{D_{\{a\}}p\Rightarrow D_{\{a,b\}}p}{\infer[(w\Rightarrow)]{p,D_{\{a\}}p\Rightarrow p}{p\Rightarrow p}}},
\]
Similarly, we can construct a derivation with $(Cut)$ of the same sequent in $\mathsf{G}(\mathbf{S5}_{D})$:

\[
\infer[(Cut)_{D_{\{a\}}p}]{\Rightarrow D_{\{a\}}\neg D_{\{a\}}p, D_{\{a,b\}}p}{\infer[(\Rightarrow D^{\mb{S5}_{D}})]{\Rightarrow D_{\{a\}}\neg D_{\{a\}}p, D_{\{a\}}p}{\infer[(\Rightarrow\neg)]{\Rightarrow \neg D_{\{a\}}p, D_{\{a\}}p}{D_{\{a\}}p\Rightarrow D_{\{a\}}p}} & \infer[(\Rightarrow D^{\mb{S5}_{D}})]{D_{\{a\}}p\Rightarrow D_{\{a,b\}}p}{\infer[(D\Rightarrow)]{D_{\{a\}}p\Rightarrow p}{p\Rightarrow p}}}.
\]

It is noted that the cut formula $D_{\{a\}} p$ is a subformula of the conclusion of the cut. Since $\Ge{\mb{S5}}$ is not cut-free, it is not surprising that $\Ge{\mb{S5}_{D}}$ is also not cut-free. In what follows, however, we show that $\Rightarrow D_{\{a\}}\neg D_{\{a\}}p, D_{\{a,b\}}p$ is {\em not} derivable in $\Ge{\mb{K45}_{D}}$, $\Ge{\mb{KD45}_{D}}$ and $\Ge{\mb{S5}_{D}}$ without the cut rule either.

 \begin{theorem}
\label{lem:unadmissible-cut}
Let $\Lambda\in\{\mb{K45}_{D},\mb{KD45}_{D},{\mb{S5}_{D}}\}$. The sequent $\Rightarrow D_{\{a\}}\neg D_{\{a\}}p, D_{\{a,b\}}p$ is {\em not} derivable in $\mathsf{G}^{-}(\Lambda)$, where $\mathsf{G}^{-}(\Lambda)$ is the same calculus as $\Ge{\Lambda}$ except that the cut rule is dropped.
 \end{theorem}
\begin{proof}
    Let $\varphi^{m}$ denote $m$ occurrences of $\varphi$. We comment on the case where $\Lambda=\mb{KD45}_{D}$. We can prove the following statements by simultaneous induction on $\mathcal{D}$ of $\mathsf{G}^{-}(\mb{KD45}_{D})$.
\begin{enumerate}[itemsep=0pt, parsep=0pt, topsep=4pt, label=(\arabic*)]
    \item $\forall m,n,l\in \mathbb{N}\ (\mathtt{root}(\mathcal{D}) \not\equiv\ \Rightarrow (D_{\{a\}}\neg D_{\{a\}}p)^{m}, (D_{\{a,b\}}p)^{n},p^{l})$,
    \item $\forall x_{1},x_{2},y,z\in \mathbb{N}\ (\mathtt{root}(\mathcal{D}) \not\equiv p^{x_{1}},(D_{\{a\}}p)^{x_{2}}\Rightarrow (D_{\{a\}}\neg D_{\{a\}}p)^{y},(\neg D_{\{a\}}p)^{z})$.
\end{enumerate}
Then the statement of the lemma is a special case of (1) where $m = n = 1$ and $l$ = $0$. 
\end{proof}

\section{Analytic Cut Property of \texorpdfstring{$\Ge{\mb{K45}_{D}}$}{G(K45)_{D}}, \texorpdfstring{$\Ge{\mb{KD45}_{D}}$}{G(KD45)_{D}} and  \texorpdfstring{$\Ge{\mb{S5}_{D}}$}{G(S5)_{D}}}
\label{sec:a-cut D}
This section establishes the analytic cut property, i.e., every application of the cut rule in $\Ge{\Lambda}$ can be replaced with an application of the cut rule such that the cut formula is in the set of {\em subformulas} of the conclusion of the cut. The set $\mathsf{Sub}(\varphi)$ of subformulas of an $L$-formula $\varphi$ is defined as usual. In particular, $\Sub(D_{G}\varphi)=\Sub(\varphi)\cup\{D_{G}\varphi\}$. For any multiset $\Gamma$ of formulas, $\Sub(\Gamma)$ is defined as $\bigcup_{\varphi\in\Gamma}\Sub(\varphi)$. Moreover, we say that $\Xi$ is {\em subformula-closed} if $\Sub(\Xi)\subseteq\Xi$. Let $\Lambda \in \{\mathbf{K45}_{D}, \mathbf{KD45}_{D}, \mathbf{S5}_{D}\}$ in this section.

\begin{definition}
\label{dfn:a-cut}
A cut
\[
\infer[(Cut)_\varphi]{\Gamma, \Pi \Rightarrow \Delta, \Sigma}{\Gamma \Rightarrow \Delta, \varphi & \varphi, \Pi \Rightarrow \Sigma}
\]
is {\em analytic} if $\varphi \in \Sub(\Gamma, \Pi, \Delta, \Sigma)$. The sequent calculus $\mathsf{G}^{\mathtt{a}}(\Lambda)$ is the same calculus as $\Ge{\Lambda}$ except that $(Cut)$ is always analytic. 
In what follows, when $\Gamma \Rightarrow \Delta$ is derivable in $\mathsf{G}^{\mathtt{a}}({\Lambda})$, we denote $\mathsf{G}^{\mathtt{a}}({\Lambda}) \vdash \Gamma \Rightarrow \Delta$.
\end{definition}

Let $\Xi$ be a subformula-closed finite set of formulas in what follows. The following definition is taken from~\cite{Takano_2018_semantical}, whose notion originates from~\cite{Takano_2001_modified}.

\begin{definition}
\label{dfn:partial valuation}
A pair $(\Gamma,\Delta)$ is a {\em $\Xi$-partial valuation} in $\mathsf{G}^{\mt{a}}(\Lambda)$ if all the following hold: 
\begin{enumerate}[label=(\arabic*)]
    \item $\mathsf{G}^{\mt{a}}(\Lambda)\not\vdash\Gamma\Rightarrow\Delta$; 
    \item $\Sub(\Gamma,\Delta)=\Gamma\cup\Delta$; 
    \item $\Gamma\cup\Delta\subseteq\Xi$.
\end{enumerate}    
\end{definition}

\begin{remark}
It is worth noting that the condition for $\Sub(\Gamma,\Delta)$ highlights the distinction between semantically eliminating only non-analytic cuts and eliminating all cuts. In fully cut-free sequent calculi (for example, $\mathbf{S4}$), the equality $\Sub(\Gamma,\Delta)=\Gamma\cup\Delta$ may \emph{not} hold; instead, the inclusion $\Gamma\cup\Delta\subseteq \Sub(\Gamma,\Delta)$ is utilized. This is precisely the difference between the approach used in this paper to establish the analytic cut property and the one applied in~\cite{Takano_2018_semantical} to prove semantic completeness for completely cut-free systems.
\end{remark}

The following lemma can be established similarly to Lindenbaum's Lemma (cf.~\cite[Lemma 1]{Udatsu_2025_craig}).

\begin{lemma}
\label{lem:Lindenbaum a-cut}
Let $\Gamma\cup\Delta\subseteq\Xi$. If $\mathsf{G}^{\mt{a}}(\Lambda)\not\vdash\Gamma\Rightarrow\Delta$, then there exists a $\Xi$-partial valuation $(\Gamma^{+},\Delta^{+})$ such that $\Gamma\subseteq\Gamma^{+}$, $\Delta\subseteq\Delta^{+}$ and $\Sub(\Gamma,\Delta)=\Gamma^{+}\cup\Delta^{+}$.
\end{lemma}

\begin{proof}
    Let $\varphi_{1},\dots,\varphi_{m}$ be an enumeration of all formulas in $\Sub(\Gamma,\Delta)$. In what follows we inductively construct a sequence $(\Gamma_{k},\Delta_{k})(1\le k\le m)$ such that $\mathsf{G}^{\mt{a}}(\Lambda)\not\vdash\Gamma_{k}\Rightarrow\Delta_{k}$, $\Gamma_{k}\subseteq\Gamma_{k+1}$ and $\Delta_{k}\subseteq\Delta_{k+1}$ for all $1\le k\le m$. For base step, let $(\Gamma_{0},\Delta_{0}):=(\Gamma,\Delta)$. It is clear that $\mathsf{G}^{\mt{a}}(\Lambda)\not\vdash\Gamma\Rightarrow\Delta$. For inductive step, we define $(\Gamma_{k+1},\Delta_{k+1})$ as follows:
    \[
    \begin{cases}
        \Gamma_{k+1}=\Gamma_{k}\text{ and }\Delta_{k+1}=\Delta_{k}\cup\{\varphi_{k}\} &\text{if }\mathsf{G}^{\mt{a}}(\Lambda)\not\vdash\Gamma_{k}\Rightarrow\Delta_{k},\varphi_{k}\\
        \Gamma_{k+1}=\Gamma_{k}\cup\{\varphi_{k}\}\text{ and }\Delta_{k+1}=\Delta_{k} &\text{otherwise}
    \end{cases}
    \]
    Moreover, it is not the case that $\mathsf{G}^{\mt{a}}(\Lambda)\vdash\Gamma_{k}\Rightarrow\Delta_{k},\varphi_{k}$ and $\mathsf{G}^{\mt{a}}(\Lambda)\vdash\varphi_{k}, \Gamma_{k}\Rightarrow\Delta_{k}$. Suppose it is the case for contradiction. Then we can derive $\Gamma_{k}\Rightarrow\Delta_{k}$ by
    \[
    \infer=[(w)]{\Gamma_{k}\Rightarrow\Delta_{k}}{\infer[(Cut)^{\mt{a}}_{\varphi_{k}}]{\Gamma_{k},\Gamma_{k}\Rightarrow\Delta_{k},\Delta_{k}}{\Gamma_{k}\Rightarrow\Delta_{k},\varphi_{k} & \varphi_{k}, \Gamma_{k}\Rightarrow\Delta_{k}}},
    \]
    where $\varphi_{k}\in\Sub(\Gamma_{k},\Delta_{k})$ holds clearly since $\varphi_{k}\in\Sub(\Gamma,\Delta)$ and $\Gamma\subseteq\Gamma_{k}$, $\Delta\subseteq\Delta_{k}$. Hence we have $\mathsf{G}^{\mt{a}}(\Lambda)\vdash\Gamma_{k}\Rightarrow\Delta_{k}$, which contradicts the induction hypothesis. Therefore, either $\Gamma_{k}\Rightarrow\Delta_{k},\varphi_{k}$ or $\varphi_{k}, \Gamma_{k}\Rightarrow\Delta_{k}$ is underivable in $\mathsf{G}^{\mt{a}}(\Lambda)$. Finally, let $(\Gamma^{+},\Delta^{+}):=(\Gamma_{m+1},\Delta_{m+1})$. It holds clearly that $\Sub(\Gamma,\Delta)=\Sub(\Gamma^{+},\Delta^{+})=\Gamma^{+}\cup\Delta^{+}$. Moreover $\Gamma\subseteq\Gamma^{+}$, $\Delta\subseteq\Delta^{+}$ and $\mathsf{G}^{\mt{a}}(\Lambda)\not\vdash\Gamma^{+}\Rightarrow\Delta^{+}$ by our definition. 
\end{proof}

The following definition of $R^{\Lambda}_{G}$ is based on the definitions of accessibility relations for individual agents given in~\cite[Section 5.3]{Maruyama_2003_temporal} for individual agents, but is modified here to the group level to accommodate the semantics of distributed knowledge.

\begin{definition}
\label{dfn:canonical pseudo model}
    The pseudo model $M^{\Lambda}_{\Xi}=(W^{\Lambda}_{\Xi}, (R^{\Lambda}_{G})_{G\in\mathsf{Grp}},V^{\Lambda}_{\Xi})$ derived from $\Xi$ is defined as follows:
    \begin{itemize}[itemsep=0pt, parsep=0pt, topsep=5pt]
        \item $W^{\Lambda}_{\Xi}=\inset{(\Gamma,\Delta)}{(\Gamma,\Delta) \text{ is a }\Xi\text{-partial valuation in }\mathsf{G}^{\mt{a}}(\Lambda)}$;
        \item $R^{\Lambda}_{G}$ is defined depending on the choice of $\Lambda$:
        \begin{itemize}
            \item $\Lambda=\mb{K45}_{D}$: $(\Gamma,\Delta)R^{\mb{K45}_{D}}_{G}(\Pi,\Sigma)$ iff 
            \begin{center}
             $\forall H\subseteq G\ (\inset{\varphi,D_{H}\varphi}{D_{H}\varphi\in\Gamma}\subseteq\Pi$ and $\inset{D_{H}\varphi}{D_{H}\varphi\in\Pi}\subseteq\Gamma$),   
            \end{center}
        \item $\Lambda=\mb{KD45}_{D}$: $R^{\mb{KD45}_{D}}_{G}:=R^{\mb{K45}_{D}}_{G}$;
        \item $\Lambda=\mb{S5}_{D}$: $(\Gamma,\Delta)R^{\mb{S5}_{D}}_{G}(\Pi,\Sigma)$ iff $\forall H\subseteq G\ (\inset{D_{H}\varphi}{D_{H}\varphi\in\Gamma}=\inset{D_{H}\varphi}{D_{H}\varphi\in\Pi})$;
        \end{itemize}
        \item $(\Gamma,\Delta)\in V^{\Lambda}_{\Xi}(p)$ iff $p\in\Gamma$.
    \end{itemize}
\end{definition}
\noindent Since $\Xi$ is finite, we note that $W^{\Lambda}_{\Xi}$ is finite. The following proposition shows that $M^{\Lambda}_{\Xi}$ is indeed a pseudo model.
\begin{proposition}
\label{prop:well definiability D}
The following hold for the pseudo frame $F^{\Lambda}_{\Xi}=(W^{\Lambda}_{\Xi}, (R^{\Lambda}_{G})_{G\in\mathsf{Grp}})$ derived from $\Xi$:
\begin{enumerate*}[1.]
    \item $G\subseteq H$ implies $R^{\Lambda}_{H}\subseteq R^{\Lambda}_{G}$; 
    \item $F^{\Lambda}_{\Xi}$ satisfies the corresponding properties of $\Lambda$.
\end{enumerate*}
\end{proposition}

\begin{proof}

\begin{enumerate}[1. ]
        \item Fix any $G,H\in\ms{Grp}$ such that $G\subseteq H$. Since $R^{\mb{KD45}_{D}}_{G}:=R^{\mb{K45}_{D}}_{G}$, we only comment on the cases where $(\Lambda =\mb{K45}_{D})$ and $(\Lambda =\mb{S5}_{D})$.
        \begin{itemize}
            \item Let $\Lambda =\mb{K45}_{D}$. Fix any $(\Gamma,\Delta)$, $(\Pi,\Sigma)\in W^{\Lambda}_{\Xi}$ such that $(\Gamma,\Delta)R^{\mb{K45}_{D}}_{H}(\Pi,\Sigma)$. We show that $(\Gamma,\Delta)R^{\mb{K45}_{D}}_{G}(\Pi,\Sigma)$. Fix any $I\subseteq G$. It suffices to show (i) $\inset{\varphi,D_{I}\varphi}{D_{I}\varphi\in\Gamma}\subseteq\Pi$ and (ii) $\inset{D_{I}\varphi}{D_{I}\varphi\in\Pi}\subseteq\Gamma$. By assumption, $I\subseteq G\subseteq H$, hence our goal follows from 
            \[
            (\Gamma,\Delta)R^{\mb{K45}_{D}}_{H}(\Pi,\Sigma).
            \]
            \item Let $\Lambda =\mb{S5}_{D}$. Fix any $(\Gamma,\Delta)$, $(\Pi,\Sigma)\in W^{\Lambda}_{\Xi}$ such that $(\Gamma,\Delta)R^{\mb{S5}_{D}}_{H}(\Pi,\Sigma)$. We show that $(\Gamma,\Delta)R^{\mb{S5}_{D}}_{G}(\Pi,\Sigma)$. Fix any $I\subseteq G$. It suffices to show $\inset{D_{I}\varphi}{D_{I}\varphi\in\Gamma}=\inset{D_{I}\varphi}{D_{I}\varphi\in\Pi}$. By assumption, $I\subseteq G\subseteq H$, hence our goal follows from $(\Gamma,\Delta)R^{\mb{S5}_{D}}_{H}(\Pi,\Sigma)$. 
        \end{itemize}
        \item \begin{itemize}[itemsep=0pt, parsep=0pt, topsep=5pt]
    \item 

Let $\Lambda =\mb{K45}_{D}$. We show that $R^{\mb{K45}_{D}}_{G}$ is transitive and Euclidean.

            \noindent \fbox{Transitivity} Fix any $(\Gamma,\Delta)$, $(\Pi,\Sigma)$, $(\Theta,\Omega)\in W^{\mb{K45}_{D}}_{\Xi}$. Suppose 
            \[(\Gamma,\Delta)R^{\mb{K45}_{D}}_{G}(\Pi,\Sigma)\text{ and } (\Pi,\Sigma)R^{\mb{K45}_{D}}_{G}(\Theta,\Omega).
            \] 
            We show $(\Gamma,\Delta)R^{\mb{K45}_{D}}_{G}(\Theta,\Omega)$. Fix any $H\subseteq G$. It suffices to show (i) $\inset{\varphi,D_{H}\varphi}{D_{H}\varphi\in\Gamma}\subseteq\Theta$ and (ii) $\inset{D_{H}\varphi}{D_{H}\varphi\in\Theta}\subseteq\Gamma$.
First, we show (i). Let us fix any $D_{H}\psi\in\Gamma$. We show $\psi,D_{H}\psi\in\Theta$. By assumption, $\inset{\varphi,D_{H}\varphi}{D_{H}\varphi\in\Gamma}\subseteq\Pi$. Then we have $D_{H}\psi\in\Pi$. Again by assumption, $\inset{\varphi,D_{H}\varphi}{D_{H}\varphi\in\Pi}\subseteq\Theta$. Therefore our goal of $\psi,D_{H}\psi\in\Theta$ follows from $D_{H}\psi\in\Pi$.  
 
Second, we move to (ii). By assumption we have 
\[\inset{D_{H}\varphi}{D_{H}\varphi\in\Theta}\subseteq\Pi\text{ and }\inset{D_{H}\varphi}{D_{H}\varphi\in\Pi}\subseteq\Gamma,
\]
which implies our goal $\inset{D_{H}\varphi}{D_{H}\varphi\in\Theta}\subseteq\Gamma$.

\noindent \fbox{Euclideanness} Fix any $(\Gamma,\Delta)$, $(\Pi,\Sigma)$, $(\Theta,\Omega)\in W^{\mb{K45}_{D}}_{\Xi}$. Suppose $(\Gamma,\Delta)R^{\mb{K45}_{D}}_{G}(\Pi,\Sigma)$ and $(\Gamma,\Delta)R^{\mb{K45}_{D}}_{G}(\Theta,\Omega)$. We show $(\Pi,\Sigma)R^{\mb{K45}_{D}}_{G}(\Theta,\Omega)$. Fix any $H\subseteq G$. It suffices to show (i) $\inset{\varphi,D_{H}\varphi}{D_{H}\varphi\in\Pi}\subseteq\Theta$ and (ii) $\inset{D_{H}\varphi}{D_{H}\varphi\in\Theta}\subseteq\Pi$.

First, we establish (i). Fix any $D_{H}\psi\in\Pi$. We show that $\psi,D_{H}\psi\in\Theta$. By $(\Gamma,\Delta)R^{\mb{K45}_{D}}_{G}(\Pi,\Sigma)$ we have $D_{H}\psi\in\Gamma$. Then our goal holds by $(\Gamma,\Delta)R^{\mb{K45}_{D}}_{G}(\Theta,\Omega)$.

Second, we move to (ii). Fix any $D_{H}\psi$ with $D_{H}\psi\in\Theta$. We show that $D_{H}\psi\in\Pi$. By $(\Gamma,\Delta)R^{\mb{K45}_{D}}_{G}(\Theta,\Omega)$, we have $D_{H}\psi\in\Gamma$. Then our goal holds by $(\Gamma,\Delta)R^{\mb{K45}_{D}}_{G}(\Pi,\Sigma)$.

\noindent \fbox{Seriality} We prove the seriality of $R_{\{a\}}^{\mb{KD45}_{D}}$. Fix any $(\Gamma,\Delta)\in W^{\mb{KD45}_{D}}_{\Xi}$. We show that there exists $(\Pi,\Sigma)\in W^{\mb{KD45}_{D}}_{\Xi}$ such that $(\Gamma,\Delta)R^{\mb{KD45}_{D}}_{\{a\}}(\Pi,\Sigma)$. We define   $\Pi:=\inset{\varphi}{D_{\{a\}}\varphi\in\Gamma}$, $\Sigma:=\inset{D_{\{a\}}\psi}{D_{\{a\}}\psi\in\Delta}$.
            Then $\Pi,D_{\{a\}}\Pi\Rightarrow\Sigma$ is underivable in $\mathsf{G}^{\mt{a}}(\mb{KD45}_{D})$. Otherwise $\Gamma\Rightarrow\Delta$ is derivable from $\Pi,D_{\{a\}}\Pi\Rightarrow\Sigma$ via $(D^{\mb{KD45}_D})$ and weakening rules, which contradicts $\mathsf{G}^{\mt{a}}(\Lambda)\not\vdash\Gamma\Rightarrow\Delta$. By Lemma~\ref{lem:Lindenbaum a-cut}, there exists $(\Pi^{+},\Sigma^{+})\in W^{\mb{KD45}_{D}}_{\Xi}$ such that $\Pi\cup D_{\{a\}}\Pi\subseteq\Pi^{+}$, $\Sigma\subseteq\Sigma^{+}$ and $\Sub(\Pi,D_{\{a\}}\Pi,\Sigma)=\Pi^{+}\cup\Sigma^{+}$. Then it suffices to show $(\Gamma,\Delta)R^{\mb{KD45}_{D}}_{\{a\}}(\Pi^{+},\Sigma^{+})$. Fix any $D_{\{a\}}\chi\in\Xi$. First, we suppose $D_{\{a\}}\chi\in\Gamma$. We show that $\chi, D_{\{a\}}\chi\in\Pi^{+}$. Then $\chi\in\Pi\subseteq\Pi^{+}$. Moreover $D_{\{a\}}\chi\in D_{\{a\}}\Pi\subseteq\Pi^{+}$. Therefore, our goal $\chi\in\Pi^{+}$ and $D_{\{a\}}\chi\in\Pi^{+}$ holds. Second, we suppose $D_{\{a\}}\chi\in\Pi$. Our goal is to show $D_{\{a\}}\chi\in\Gamma$. Suppose $D_{\{a\}}\chi\notin\Gamma$ for contradiction. By $D_{\{a\}}\chi\in\Pi$, $D_{\{a\}}\chi\in\Pi^{+}\cup\Sigma^{+}=\Sub(\Pi,D_{\{a\}}\Pi, \Sigma)\subseteq\Sub(\Gamma,\Delta)=\Gamma\cup\Delta$. Since $D_{\{a\}}\chi\notin\Gamma$, we have $D_{\{a\}}\chi\in\Delta$, i.e., $D_{\{a\}}\chi\in\Sigma\subseteq\Sigma^{+}$. Then $\Pi^{+}\Rightarrow\Sigma^{+}$ is derivable from $D_{\{a\}}\chi\Rightarrow D_{\{a\}}\chi$ by weakening rules, which is contradicted to $(\Pi^{+},\Sigma^{+})$ being a $\Xi$-partial valuation. Therefore, our goal $D_{\{a\}}\chi\in\Gamma$ holds. 
        
        \item Let $\Lambda$ = $ \mb{S5}_{D}$. We show that $R^{\mb{S5}_{D}}_{G}$ is reflexive, symmetric and transitive. Fix any $(\Gamma,\Delta)$, $(\Pi,\Sigma)$, $(\Theta,\Omega)\in W^{\mb{S5}_{D}}_{\Xi}$. First, we show that $(\Gamma,\Delta)R^{\mb{S5}_{D}}_{G}(\Gamma,\Delta)$, which holds clearly by 
        \[\inset{D_{H}\varphi}{D_{H}\varphi\in\Gamma}=\inset{D_{H}\varphi}{D_{H}\varphi\in\Gamma}.
        \]
        Second, suppose $(\Gamma,\Delta)R^{\mb{S5}_{D}}_{G}(\Pi,\Sigma)$. We show that $(\Pi,\Sigma)R^{\mb{S5}_{D}}_{G}(\Gamma,\Delta)$. It suffices to show that $\inset{D_{H}\varphi}{D_{H}\varphi\in\Pi}=\inset{D_{H}\varphi}{D_{H}\varphi\in\Gamma}$, which follows from our assumption 
        \[(\Gamma,\Delta)R^{\mb{S5}_{D}}_{G}(\Pi,\Sigma),
        \]
        i.e., $\inset{D_{H}\varphi}{D_{H}\varphi\in\Gamma}=\inset{D_{H}\varphi}{D_{H}\varphi\in\Pi}$. Finally, suppose 
        \[(\Gamma,\Delta)R^{\mb{S5}_{D}}_{G}(\Pi,\Sigma)\text{ and }(\Pi,\Sigma)R^{\mb{S5}_{D}}_{G}(\Theta,\Omega).
        \]
        We show $(\Gamma,\Delta)R^{\mb{S5}_{D}}_{G}(\Theta,\Omega)$. Fix any $H\subseteq G$. It suffices to show $\inset{D_{H}\varphi}{D_{H}\varphi\in\Gamma}=\inset{D_{H}\varphi}{D_{H}\varphi\in\Theta}$. By assumption, we have $\inset{D_{H}\varphi}{D_{H}\varphi\in\Gamma}=\inset{D_{H}\varphi}{D_{H}\varphi\in\Pi}$ and $\inset{D_{H}\varphi}{D_{H}\varphi\in\Pi}=\inset{D_{H}\varphi}{D_{H}\varphi\in\Theta}$. Therefore, our goal holds.\qedhere
        \end{itemize}
    \end{enumerate}
\end{proof}

\begin{lemma}
\label{lem:saturated}
Let $(\Gamma,\Delta)$ be a $\Xi$-partial valuation in $\mathsf{G}^{\mt{a}}(\Lambda)$. Then $(\Gamma,\Delta)$ satisfies the following:
\begin{itemize}[itemsep=0pt, parsep=0pt, topsep=5pt]
    \item[$(\neg\Rightarrow)$] $\neg\varphi\in\Gamma$ implies $\varphi\in\Delta$;
    \item[$(\Rightarrow\neg)$] $\neg\varphi\in\Delta$ implies $\varphi\in\Gamma$;
    \item[$(\to\Rightarrow)$] $\varphi\to\psi\in\Gamma$ implies $\varphi\in\Delta$ or $\psi\in\Gamma$;
    \item[$(\Rightarrow\to)$] $\varphi\to\psi\in\Delta$ implies $\varphi\in\Gamma$ and $\psi\in\Delta$;
    \item[$(D)$] $D_{G}\varphi\in\Delta$ implies $\varphi\in\Sigma$ for some $(\Pi,\Sigma)\in W^{\Lambda}_{\Xi}$ such that $(\Gamma,\Delta)R^{\Lambda}_{G}(\Pi,\Sigma)$.
\end{itemize}
\end{lemma}
\begin{proof}
We comment on item ($D$) for $\mb{K45}_{D}$ and $\mb{S5}_{D}$, because the other conditions are established easily (cf.~~\cite[Lemma 2]{Udatsu_2025_craig}). 
    \begin{itemize}[itemsep=0pt, parsep=0pt, topsep=5pt]
        \item Let $\Lambda=\mb{K45}_{D}$. Suppose $D_{G}\varphi\in\Delta$. Define 
    $\Pi:=\inset{\psi,D_{H}\psi}{D_{H}\psi\in\Gamma \text{ and }H\subseteq G}$ and\\  
    $\Sigma:=\inset{D_{I}\gamma}{D_{I}\gamma\in\Delta\text{ and }I\subseteq G}$.
    It follows that $\Pi\Rightarrow \Sigma,\varphi$ is underivable in $\mathsf{G}^{\mt{a}}(\mb{K45}_{D})$. Suppose not. 
    Then we obtain the following:
    \[
    \infer[(\Rightarrow D^{\mb{K45}_{D}})]{\Pi'\Rightarrow \Sigma,D_{G}\varphi}{\Pi\Rightarrow \Sigma,\varphi},
    \]
    where $\Pi'= \bigcup_{H \subseteq G}\inset{D_{H}\psi}{D_{H}\psi\in\Gamma}$, 
    which is contradicted to $(\Gamma,\Delta)$ being a $\Xi$-partial valuation by weakening rules. 
    So, $\mathsf{G}^{\mt{a}}(\mb{K45}_{D})\not\vdash\Pi\Rightarrow \Sigma,\varphi$. By Lemma~\ref{lem:Lindenbaum a-cut}, there exists a $\Xi$-partial valuation $(\Pi^{+},\Sigma^{+})$ such that $\Pi\subseteq\Pi^{+}$, $\Sigma\cup\{\varphi\}\subseteq\Sigma^{+}$ and $\Sub(\Pi, \Sigma,\varphi)=\Pi^{+}\cup\Sigma^{+}$. Then it suffices to show $(\Gamma,\Delta)R_{G}^{\mb{K45}_{D}}(\Pi^{+},\Sigma^{+})$, which can be proved similarly to the proof of Proposition~\ref{prop:well definiability D} (2). 
    \item Let $\Lambda=\mb{S5}_{D}$. Suppose $D_{G}\varphi\in\Delta$.
   Let $\Pi$ and $\Sigma$ be defined as:
    \[
    \Pi:=\inset{D_{H}\psi}{D_{H}\psi\in\Gamma \text{ and }H\subseteq G}\quad 
    \Sigma:=\inset{D_{I}\gamma}{D_{I}\gamma\in\Delta\text{ and }I\subseteq G}.
    \]
    It follows that $\Pi\Rightarrow\Sigma,\varphi$ is underivable in $\mathsf{G}^{\mt{a}}(\mb{S5}_{D})$, otherwise $\Gamma\Rightarrow\Delta$ is derivable by $(\Rightarrow D^{\mb{S5}_{D}})$ and weakening rules. Then by Lemma~\ref{lem:Lindenbaum a-cut}, there exists a $\Xi$-partial valuation $(\Pi^{+},\Sigma^{+})$ such that $\Pi\subseteq\Pi^{+}$, $\Sigma\cup\{\varphi\}\subseteq\Sigma^{+}$ and $\Sub(\Pi, \Sigma,\varphi)=\Pi^{+}\cup\Sigma^{+}$. Then it suffices to show $(\Gamma,\Delta)R_{G}^{\mb{S5}_{D}}(\Pi^{+},\Sigma^{+})$. Fix any $C\subseteq G$ and $\chi\in\Xi$. We show that $D_{C}\chi\in\Gamma$ iff $D_{C}\chi\in\Pi^{+}$. 
    ($\Rightarrow$) Suppose $D_{C}\chi\in\Gamma$. By definition of $\Pi$, $D_{C}\chi\in\Pi$. Therefore our goal $D_{C}\chi\in\Pi^{+}$ holds by $\Pi\subseteq\Pi^{+}$. ($\Leftarrow$) Suppose $D_{C}\chi\in\Pi^{+}$. Then $D_{C}\chi\in\Pi^{+}\cup\Sigma^{+}=\Sub(\Pi,\Sigma,\varphi)\subseteq\Sub(\Gamma,\Delta)=\Gamma\cup\Delta$. Our goal is to show that $D_{C}\chi\in\Gamma$. Suppose not. Then $D_{C}\chi\in\Delta$, i.e., $D_{C}\chi\in\Sigma\subseteq\Sigma^{+}$. Then $\Pi^{+}\Rightarrow\Sigma^{+}$ is derivable from $D_{C}\varphi\Rightarrow D_{C}\varphi$. Therefore our goal $D_{C}\chi\in\Gamma$ holds. \qedhere
    \end{itemize}
\end{proof}

By Lemma~\ref{lem:saturated}, we obtain the following. 

\begin{lemma}
\label{lem:truth lemma-D}
The following holds for all formulas $\varphi\in\Xi$ and $(\Gamma,\Delta)\in W^{\Lambda}_{\Xi}$;
\begin{center}
\begin{enumerate*}[itemsep=0pt, parsep=0pt, topsep=5pt, label=\arabic*., itemjoin={\quad}]
    \item $\varphi\in\Gamma$ implies $M^{\Lambda}_{\Xi},(\Gamma,\Delta)\models\varphi$; 
    \item $\varphi\in\Delta$ implies $M^{\Lambda}_{\Xi},(\Gamma,\Delta)\not\models\varphi$.
\end{enumerate*}
\end{center}
\end{lemma}

\begin{proof}
    By induction on $\varphi$. We comment on the case where $\varphi\equiv D_{G}\psi$.
        \begin{enumerate}[itemsep=0pt, parsep=0pt, topsep=5pt]
        \item Suppose $D_{G}\psi\in\Gamma$. We show that $M^{\Lambda}_{\Xi},(\Gamma,\Delta)\models D_{G}\psi$. Fix any $(\Pi,\Sigma)\in W^{\Lambda}_{\Xi}$ and suppose 
        \[(\Gamma,\Delta)R_{G}^{\Lambda}(\Pi,\Sigma).\] We show that $M^{\Lambda}_{\Xi},(\Pi,\Sigma)\models\psi$. We comment on the case for $\mb{K45}_{D}$ and $\mb{S5}_{D}$.
        \begin{itemize}
            \item $(\mb{K45}_{D})$ We show $M^{\mb{K45}_{D}}_{\Xi},(\Pi,\Sigma)\models\psi$. By induction hypothesis, it suffices to show $\psi\in\Pi$, which follows from Definition~\ref{dfn:canonical pseudo model} and $D_{G}\psi\in\Gamma$.
            \item $(\mb{S5}_{D})$ We show $M^{\mb{S5}_{D}}_{\Xi},(\Gamma,\Delta)\models D_{G}\psi$. Fix any $(\Pi,\Sigma)\in W^{\mb{S5}_{D}}_{\Xi}$ and suppose $(\Gamma,\Delta)R_{G}^{\mb{S5}_{D}}(\Pi,\Sigma)$. Then $D_{G}\psi\in\Gamma$ implies $D_{G}\psi\in\Pi$. Then $\psi\in\Sub(D_{G}\psi)\subseteq\Sub(\Pi,\Sigma)=\Pi\cup\Sigma$. If $\psi\in\Pi$, our goal $M^{\mb{K45}_{D}}_{\Xi},(\Pi,\Sigma)\models\psi$ follows from induction hypothesis. So suppose $\psi\notin\Pi$ for contradiction. Then $\psi\in\Sigma$, hence $\Pi\Rightarrow\Sigma$ is derivable by
        \[
        \infer=[(w)]{\Pi\Rightarrow\Sigma}{\infer[(D\Rightarrow)]{D_{G}\varphi\Rightarrow\varphi}{\varphi\Rightarrow\varphi}},
        \]
        which contradicts $(\Pi,\Sigma)$ being a $\Xi$-partial valuation. Therefore our goal $\psi\in\Pi$ holds.
        \end{itemize}
         \item Suppose $D_{G}\psi\in\Delta$. We show $M^{\Lambda}_{\Xi},(\Gamma,\Delta)\not\models D_{G}\psi$. By Lemma~\ref{lem:saturated}, we can find some $(\Pi,\Sigma)\in W^{\Lambda}_{\Xi}$ such that $(\Gamma,\Delta)R_{G}^{\Lambda}(\Pi,\Sigma)$ and $\psi\in\Sigma$. By induction hypothesis, $\psi\in\Sigma$ implies $M^{\Lambda}_{\Xi},(\Pi,\Sigma)\not\models\psi$. Therefore our goal $M^{\Lambda}_{\Xi},(\Gamma,\Delta)\not\models D_{G}\psi$ holds.\qedhere
    \end{enumerate}
\end{proof}

\begin{theorem}[Completeness of $\mathsf{G}^{\mt{a}}(\Lambda)$]
\label{thm:completeness-D}
For any sequent $\Gamma\Rightarrow\Delta$, if $\bigwedge\Gamma\to\bigvee\Delta$ is valid in {\em all} pseudo frames which satisfies the corresponding properties to $\Lambda$, then $\mathsf{G}^{\mt{a}}(\Lambda)\vdash\Gamma\Rightarrow\Delta$.
\end{theorem}
\begin{proof}
    We show the contraposition. Suppose $\mathsf{G}^{\mt{a}}(\Lambda)\not\vdash\Gamma\Rightarrow\Delta$. Let $\Xi:=\Sub(\Gamma,\Delta)$. By Lemma~\ref{lem:Lindenbaum a-cut}, there exists a $\Xi$-partial valuation $(\Gamma^{+},\Delta^{+})$ such that $\Gamma\subseteq\Gamma^{+}$, $\Delta\subseteq\Delta^{+}$ and $\Sub(\Gamma,\Delta)=\Gamma^{+}\cup\Delta^{+}$. By Lemma~\ref{lem:truth lemma-D}, $M^{\Lambda}_{\Xi},(\Gamma^{+},\Delta^{+})\not\models\bigwedge\Gamma\to\bigvee\Delta$. Therefore our goal holds.
\end{proof}

As a corollary of Theorem~\ref{thm:completeness-D}, we can show that $\Ge{\Lambda}$ enjoys the analytic cut property.

\begin{corollary}
\label{cor:a-cut}
If a sequent $\Gamma \Rightarrow \Delta$ is derivable in $\Ge{\Lambda}$, then it is also derivable in  $\mathsf{G}^{\mathtt{a}}(\Lambda)$.    
\end{corollary}
\begin{proof}
Suppose $\Ge{\Lambda}\vdash\Gamma \Rightarrow \Delta$. By Proposition~\ref{prop:soundness-D}, $\bigwedge\Gamma\to\bigvee\Delta$ is valid in $\mathbb{F}_{\Lambda}$. Therefore, our goal $\mathsf{G}^{\mt{a}}(\Lambda)\vdash\Gamma\Rightarrow\Delta$ holds by Theorem~\ref{thm:completeness-D}.
\end{proof}


\section{Craig Interpolation Theorem}
\label{sec:CIT}
This section establishes the Craig interpolation theorem for $\Ge{\Lambda}$ using the analytic cut property via Maehara's method, originally introduced in~\cite{Maehara_1961_craig}. An application of this method to basic modal logic can also be found in~\cite{Ono_1998_prooftheoretic}. First, we introduce some necessary syntactic notions. As in the previous section, we assume in what follows that $\Lambda \in \{\mathbf{K45}_{D}, \mathbf{KD45}_{D}, \mathbf{S5}_{D}\}$.

\begin{definition}
A {\em partition} for a sequent $\Gamma\Rightarrow\Delta$ is defined as a tuple $\langle(\Gamma_{1}:\Delta_{1});(\Gamma_{2}:\Delta_{2})\rangle$ such that $\Gamma=\Gamma_{1}, \Gamma_{2}$ and $\Delta=\Delta_{1},\Delta_{2}$.
\end{definition}

\begin{definition}
\label{dfn:Prop-Agt-D}
For a formula $\varphi\in L$, we define $\mt{Prop}(\varphi)$ as the set of all propositional variables appearing in $\varphi$ and $\Agt{\varphi}$ as the set of all agents appearing in $\varphi$. For any multiset $\Gamma$ of formulas, $\Prop{\Gamma}$ and $\Agt{\Gamma}$ are defined as $\bigcup_{\varphi\in\Gamma}\Prop{\varphi}$ and $\bigcup_{\varphi\in\Gamma}\Agt{\varphi}$, respectively.
\end{definition}

In particular, it is noted that $\Agt{D_{G}\varphi}  =\Agt{\varphi}\cup G$. The following lemma is immediate.

\begin{lemma}
\label{lem:craig-sub-D}
If $\varphi\in\Sub(\psi)$, then $\Prop{\varphi}\subseteq\Prop{\psi}$ and $\Agt{\varphi}\subseteq\Agt{\psi}$.
\end{lemma}

The following is the key lemma for the Craig interpolation theorem.

\begin{lemma}
\label{lem:craig-D}
Suppose $\Ge{\Lambda}\vdash\Gamma\Rightarrow\Delta$. Then for any partition $\langle(\Gamma_{1}:\Delta_{1});(\Gamma_{2}:\Delta_{2})\rangle$ for the sequent $\Gamma\Rightarrow\Delta$, there exists a formula $\varphi$ $($called {\em interpolant}$)$ that satisfies the following:
\begin{itemize}
    \item[$(1)$] $\Ge{\Lambda}\vdash\Gamma_{1}\Rightarrow\Delta_{1},\varphi$ and $\Ge{\Lambda}\vdash\varphi, \Gamma_{2}\Rightarrow\Delta_{2}$; 
    \item[$(2)$] $\Prop{\varphi}\subseteq\Prop{\Gamma_{1},\Delta_{1}}\cap\Prop{\Gamma_{2},\Delta_{2}}$; 
    \item[$(3)$] $\Agt{\varphi}\subseteq\Agt{\Gamma_{1},\Delta_{1}}\cap\Agt{\Gamma_{2},\Delta_{2}}$.
\end{itemize}
\end{lemma}
\begin{proof}
    In what follows, we let $\Lambda$ be $\mathbf{K45}_{D}$. This is because the cases where $\Lambda$ is $\mathbf{KD45}_{D}$ or $\mathbf{S5}_{D}$ can be handled using an argument similar to that for $\mathbf{K45}_{D}$. By Corollary~\ref{cor:a-cut}, we can assume that $\Gamma\Rightarrow\Delta$ has a derivation $\D$ in $\mathsf{G}^{\mathtt{a}}(\Lambda)$. We prove Lemma~\ref{lem:craig-D} by induction on $\D$. Let the height of $\D$ be $n$.  For the base case, our argument is standard (the reader is referred to~\cite[Lemma 36]{Ono_1998_prooftheoretic}). For  the inductive step, let $n>0$. First, we suppose $\Gamma_{1}=\Delta_{1}=\emptyset$ or $\Gamma_{2}=\Delta_{2}=\emptyset$. Then the interpolant is $\neg\bot$ or $\bot$. Hence, in what follows, we suppose it is not the case that $\Gamma_{1}=\Delta_{1}=\emptyset$ or $\Gamma_{2}=\Delta_{2}=\emptyset$, i.e., $\Gamma_{1}\cup\Delta_{1}\ne\emptyset$ and $\Gamma_{2}\cup\Delta_{2}\ne\emptyset$. We divide our argument depending on the last applied rule to obtain $\Gamma\Rightarrow\Delta$. We comment on the cases where the last applied rule is an analytic cut or $(\Rightarrow D^{\mathbf{K45}_{D}})$. Although an argument for the case of analytic cut is discussed already in~\cite[Section~6.4]{Ono_1998_prooftheoretic}, we provide the details here in order to make the proof self-contained.
    For other cases, the reader is referred to~\cite[Lemma 34]{Ono_1998_prooftheoretic} for propositional connectives.
    \begin{itemize}
    \item Suppose that the last applied rule is an analytic cut:
\[
\infer[(Cut)_\varphi]{\Gamma_{1},\Gamma_{2}, \Pi_{1},\Pi_{2} \Rightarrow \Delta_{1},\Delta_{2}, \Sigma_{1},\Sigma_{2}}{\Gamma_{1},\Gamma_{2} \Rightarrow \Delta_{1},\Delta_{2}, \varphi & \varphi, \Pi_{1},\Pi_{2} \Rightarrow \Sigma_{1},\Sigma_{2}},
\]
    where $\varphi\in\Sub(\Gamma_{1},\Gamma_{2}, \Pi_{1},\Pi_{2}, \Delta_{1},\Delta_{2}, \Sigma_{1},\Sigma_{2})$. The partition is of the form $$\Part{\Gamma_{1},\Pi_{1}:\Delta_{1},\Sigma_{1}}{\Gamma_{2},\Pi_{2}:\Delta_{2},\Sigma_{2}}.$$ 
    We divide the argument into the following cases:
    \[\varphi\in\Sub(\Gamma_{1},\Pi_{1},\Delta_{1},\Sigma_{1})\text{ and }\varphi\in\Sub(\Gamma_{2},\Pi_{2},\Delta_{2},\Sigma_{2}).
    \]
    We focus only on the former case.    
    Suppose $\varphi\in\Sub(\Gamma_{1},\Pi_{1},\Delta_{1},\Sigma_{1})$. By induction hypothesis, we can find formulas $\psi_{1}$ and $\psi_{2}$ such that 
    \begin{itemize}
        \item both $\Gamma_{1},\Rightarrow\Delta_{1},\varphi,\psi_{1}$ and $\psi_{1},\Gamma_{2}\Rightarrow\Delta_{2}$ are derivable in $\Ge{\Lambda}$ and $\X{\psi_{1}}\subseteq\X{\Gamma_{1},\Delta_{1},\varphi}\cap\X{\Gamma_{2},\Delta_{2}}$ for all $\mathtt{X}\in \setof{\mathsf{Prop}, \mathsf{Agt}}$; 
        \item both $\varphi,\Pi_{1},\Rightarrow\Sigma_{1},\psi_{2}$ and $\psi_{2},\Pi_{2}\Rightarrow\Sigma_{2}$ are derivable in $\Ge{\Lambda}$ and $\X{\psi_{2}}\subseteq\X{\varphi,\Pi_{1},\Sigma_{1}}\cap\X{\Pi_{2},\Sigma_{2}}$ for all $\mathtt{X}\in \setof{\mathsf{Prop}, \mathsf{Agt}}$.  
    \end{itemize}  
     We show that $\psi_{1}\lor\psi_{2}:=\neg\psi_{1}\to\psi_{2}$ is a desired interpolant. The derivability condition is obtained as follows: 
    \[
    \infer[(\Rightarrow\to)]{\Gamma_{1},\Pi_{1}\Rightarrow\Delta_{1},\Sigma_{1},\neg\psi_{1}\to\psi_{2}}{\infer[(\neg\Rightarrow)]{\neg\psi_{1},\Gamma_{1},\Pi_{1}\Rightarrow\Delta_{1},\Sigma_{1},\psi_{2}}{\infer[(Cut)^{\mt{a}}_{\varphi}]{\Gamma_{1},\Pi_{1}\Rightarrow\Delta_{1},\Sigma_{1},\psi_{1},\psi_{2}}{\Gamma_{1},\Rightarrow\Delta_{1},\varphi,\psi_{1} & \varphi,\Pi_{1},\Rightarrow\Sigma_{1},\psi_{2}}}},
    \quad
    \infer[(\to\Rightarrow)]{\neg\psi_{1}\to\psi_{2},\Gamma_{2},\Pi_{2}\Rightarrow\Delta_{2},\Sigma_{2}}{\infer[(\Rightarrow\neg)]{\Gamma_{2}\Rightarrow\Delta_{2},\neg\psi_{1}}{\psi_{1},\Gamma_{2}\Rightarrow\Delta_{2}} & \psi_{2},\Pi_{2}\Rightarrow\Sigma_{2}},
    \]
    where the $(Cut)$ is analytic by our assumption $\varphi\in\Sub(\Gamma_{1},\Pi_{1},\Delta_{1},\Sigma_{1})$.
    Moreover the variable and agent conditions are verified as follows. 
    Let $\mathsf{X} \in \{ \mathsf{Prop}, \mathsf{Agt}\}$. We proceed as follows: 
    \begin{align*}   \X{\neg\psi_{1}\to\psi_{2}}&=\X{\psi_{1}}\cup\X{\psi_{2}},\\
    &\subseteq(\X{\Gamma_{1},\Delta_{1},\varphi}\cap\X{\Gamma_{2},\Delta_{2}})\cup(\X{\varphi,\Pi_{1},\Sigma_{1}}\cap\X{\Pi_{2},\Sigma_{2}})\\
    &\subseteq\X{\Gamma_{1},\Pi_{1},\Delta_{1},\Sigma_{1},\varphi}\cap\X{\Gamma_{2},\Pi_{2},\Delta_{2},\Sigma_{2}}.
    \end{align*}
    By Lemma~\ref{lem:craig-sub-D}, $\varphi\in\Sub(\Gamma_{1},\Pi_{1},\Delta_{1},\Sigma_{1})$ implies $\X{\varphi}\subseteq\X{\Gamma_{1},\Pi_{1},\Delta_{1},\Sigma_{1}}$.  Therefore, we obtain $\X{\neg\psi_{1}\to\psi_{2}}\subseteq\X{\Gamma_{1},\Pi_{1},\Delta_{1},\Sigma_{1}}\cap\X{\Gamma_{2},\Pi_{2},\Delta_{2},\Sigma_{2}}$. 
   \item Suppose that the last applied rule is $(\Rightarrow D^{\mb{K45}_{D}})$:
    \[
    {
\infer[(\Rightarrow D^{\mb{K45}_{D}})]{\overrightarrow{D_{G_{1x}}\varphi_{1x}},\overrightarrow{D_{G_{2y}}\varphi_{2y}}\Rightarrow \overrightarrow{D_{H_{1a}}\psi_{1a}},\overrightarrow{D_{H_{2b}}\psi_{2b}}, D_{G}\chi}{\overrightarrow{\varphi_{1x}},\overrightarrow{\varphi_{2y}},\overrightarrow{D_{G_{1x}}\varphi_{1x}},\overrightarrow{D_{G_{2y}}\varphi_{2y}}\Rightarrow \overrightarrow{D_{H_{1a}}\psi_{1a}},\overrightarrow{D_{H_{2b}}\psi_{2b}}, \chi}}{},
\]
where $\overrightarrow{\gamma_{iz}}$ denotes a list $\gamma_{i1},\dots,\gamma_{im_{z}}$ of formulas and $G_{1x}$, $H_{1a}$, $G_{2y}$, $H_{2b} \subseteq G$ hold for all indices. 
There are the following two possible forms of the partitions of 
the conclusion of the rule:
\begin{small}
   \begin{enumerate}[itemsep=0pt, parsep=0pt, topsep=5pt, label=(\alph*)]
    \item $\Part{\overrightarrow{D_{G_{1x}}\varphi_{1x}}:\overrightarrow{D_{H_{1a}}\psi_{1a}},D_{G}\chi}{\overrightarrow{D_{G_{2y}}\varphi_{2y}}:\overrightarrow{D_{H_{2b}}\psi_{2b}}}$,
    \item $\Part{\overrightarrow{D_{G_{1x}}\varphi_{1x}}:\overrightarrow{D_{H_{1a}}\psi_{1a}}}{\overrightarrow{D_{G_{2y}}\varphi_{2y}}:\overrightarrow{D_{H_{2b}}\psi_{2b}},D_{G}\chi}$.
\end{enumerate} 
\end{small}
Due to space limitations, we focus on case (a) alone. By induction hypothesis, we can find a formula $\psi$ such that both $\overrightarrow{\varphi_{1x}},\overrightarrow{D_{G_{1x}}\varphi_{1x}}\Rightarrow \overrightarrow{D_{H_{1a}}\psi_{1a}},\chi,\psi$ and $\psi, \overrightarrow{\varphi_{2y}},\overrightarrow{D_{G_{2y}}\varphi_{2y}}\Rightarrow \overrightarrow{D_{H_{2b}}\psi_{2b}}$ are derivable in $\mathsf{G}^{\mathtt{a}}(\mb{K45}_{D})$ and $\psi$ satisfies $\mathsf{X}(\psi)\subseteq\mathsf{X}(\overrightarrow{\varphi_{1x}},\overrightarrow{D_{G_{1x}}\varphi_{1x}},\overrightarrow{D_{H_{1a}}\psi_{1a}},\chi)\cap\mathsf{X}(\overrightarrow{D_{G_{2y}}\varphi_{2y}},\overrightarrow{D_{H_{2b}}\psi_{2b}})$, where $\mathtt{X}\in \setof{\mathsf{Prop}, \mathsf{Agt}}$. We show that $\neg D_{H}\neg\psi$ is a desired interpolant, where $H=\bigcup \overrightarrow{G_{2y}}\cup \bigcup \overrightarrow{H_{2b}}$. We note that $H\ne\emptyset$ since we have assumed $\Gamma_{2}\cup\Delta_{2}\ne\emptyset$. It follows that $H\subseteq G$.
\begin{small}
\[
\infer[(\Rightarrow\neg)]{\overrightarrow{D_{G_{1x}}\varphi_{1x}}\Rightarrow \overrightarrow{D_{H_{1a}}\psi_{1a}},D_{G}\chi,\neg D_{H}\neg\psi}{\infer[(\Rightarrow D^{\mb{K45}_{D}})]{D_{H}\neg\psi,\overrightarrow{D_{G_{1x}}\varphi_{1x}}\Rightarrow \overrightarrow{D_{H_{1a}}\psi_{1a}},D_{G}\chi}{\infer[(w\Rightarrow)]{\neg\psi,D_{H}\neg\psi,\overrightarrow{\varphi_{1x}},\overrightarrow{D_{G_{1x}}\varphi_{1x}}\Rightarrow \overrightarrow{D_{H_{1a}}\psi_{1a}},\chi}{\infer[(\neg\Rightarrow)]{\neg\psi,\overrightarrow{\varphi_{1x}},\overrightarrow{D_{G_{1x}}\varphi_{1x}}\Rightarrow \overrightarrow{D_{H_{1a}}\psi_{1a}},\chi}{\overrightarrow{\varphi_{1x}},\overrightarrow{D_{G_{1x}}\varphi_{1x}}\Rightarrow \overrightarrow{D_{H_{1a}}\psi_{1a}},\chi,\psi}}}}
\quad
\infer[(\neg\Rightarrow)]{\neg D_{H}\neg\psi,\overrightarrow{D_{G_{2y}}\varphi_{2y}}\Rightarrow \overrightarrow{D_{H_{2b}}\psi_{2b}}}{\infer[(\Rightarrow D^{\mb{K45}_{D}})]{\overrightarrow{D_{G_{2y}}\varphi_{2y}}\Rightarrow \overrightarrow{D_{H_{2b}}\psi_{2b}},D_{H}\neg\psi}{\infer[(\Rightarrow\neg)]{\overrightarrow{\varphi_{2y}},\overrightarrow{D_{G_{2y}}\varphi_{2y}}\Rightarrow \overrightarrow{D_{H_{2b}}\psi_{2b}},\neg\psi}{\psi,\overrightarrow{\varphi_{2y}},\overrightarrow{D_{G_{2y}}\varphi_{2y}}\Rightarrow \overrightarrow{D_{H_{2b}}\psi_{2b}}}}},
\]
\end{small}
where both applications of $(\Rightarrow D^{\mb{K45}_{D}})$ are eligible. Since the condition for propositional variables is easily satisfied, we comment on the condition for agents. Our goal is to show that $\Agt{\neg D_{H}\neg\psi}\subseteq\mathsf{Agt}(\overrightarrow{D_{G_{1x}}\varphi_{1x}},\overrightarrow{D_{H_{1a}}\psi_{1a}},D_{G}\chi)\cap\mathsf{Agt}(\overrightarrow{D_{G_{2y}}\varphi_{2y}},\overrightarrow{D_{H_{2b}}\psi_{2b}})$. However, note that this argument is {\em not} possible when $\Gamma_{2}=\Delta_{2}=\emptyset$, in which case $\mathsf{Agt}(\overrightarrow{D_{G_{2y}}\varphi_{2y}},\overrightarrow{D_{H_{2b}}\psi_{2b}})=\emptyset$. The same situation happens when $\Gamma_{1}=\Delta_{1}=\emptyset$ for case (b). This is precisely why we made the assumption $\Gamma_{1}=\Delta_{1}=\emptyset$ or $\Gamma_{2}=\Delta_{2}=\emptyset$ in the first place. Fix any $a\in \Agt{\neg D_{H}\neg\psi}=H\cup\Agt{\psi}$. We divide the argument into the following two cases: $a\in \Agt{\psi}$ and $a\notin \Agt{\psi}$.
\begin{itemize}[itemsep=0pt, parsep=0pt, topsep=5pt]
    \item Suppose $a\in \Agt{\psi}$. Then our goal 
    \[a\in \mathsf{Agt}(\overrightarrow{D_{G_{1x}}\varphi_{1x}},\overrightarrow{D_{H_{1a}}\psi_{1a}},D_{G}\chi)\cap\mathsf{Agt}(\overrightarrow{D_{G_{2y}}\varphi_{2y}},\overrightarrow{D_{H_{2b}}\psi_{2b}})
    \] 
    follows from $\Agt{\psi}\subseteq\mathsf{Agt}(\overrightarrow{\varphi_{1x}},\overrightarrow{D_{G_{1x}}\varphi_{1x}},\overrightarrow{D_{H_{1a}}\psi_{1a}},\chi)\cap\mathsf{Agt}(\overrightarrow{D_{G_{2y}}\varphi_{2y}},\overrightarrow{D_{H_{2b}}\psi_{2b}})$.  
    Note that it is possible that $\{\overrightarrow{\varphi_{1x}}\}=\{\overrightarrow{D_{G_{1x}}\varphi_{1x}}\}=\emptyset$, in which case the argument remains valid.
    \item Suppose $a\notin \Agt{\psi}$. Then $a\in H$. Our goal is to show 
    \[
    a\in \mathsf{Agt}(\overrightarrow{D_{G_{1x}}\varphi_{1x}},\overrightarrow{D_{H_{1a}}\psi_{1a}},D_{G}\chi)\cap\mathsf{Agt}(\overrightarrow{D_{G_{2y}}\varphi_{2y}},\overrightarrow{D_{H_{2b}}\psi_{2b}}),
    \]
    i.e., $a\in \mathsf{Agt}(\overrightarrow{D_{G_{1x}}\varphi_{1x}},\overrightarrow{D_{H_{1a}}\psi_{1a}},D_{G}\chi)$ and $a\in \mathsf{Agt}(\overrightarrow{D_{G_{2y}}\varphi_{2y}},\overrightarrow{D_{H_{2b}}\psi_{2b}})$. The first condition is satisfied by $a\in H\subseteq G\subseteq \Agt{D_{G}\chi}$, and the second condition is satisfied by $H=\bigcup \overrightarrow{G_{2y}}\cup \bigcup \overrightarrow{H_{2b}}$. \qedhere
\end{itemize}
\end{itemize}
\end{proof}

\begin{theorem}
\label{thm:craig-D}
Suppose $\Ge{\Lambda}\vdash\Rightarrow\varphi\to\psi$. Then there exists a formula $\chi$ satisfying the following:
    \begin{itemize}
        \item[$(1)$] $\Ge{\Lambda}\vdash\Rightarrow\varphi\to\chi$ and $\Ge{\Lambda}\vdash\Rightarrow\chi\to\psi$;
        \item[$(2)$]  $\Prop{\chi}\subseteq\Prop{\varphi}\cap\Prop{\psi}$;
        \item[$(3)$]  $\Agt{\chi}\subseteq\Agt{\varphi}\cap\Agt{\psi}$.
    \end{itemize}
\end{theorem}

\begin{proof}
Suppose $\Ge{\Lambda}\vdash\Rightarrow\varphi\to\psi$. Then we can show $\Ge{\Lambda}\vdash\varphi\Rightarrow\psi$ by
\[
\infer[(Cut)_{\varphi\to\psi}]{\varphi\Rightarrow\psi}{\Rightarrow\varphi\to\psi & \infer[(\to\Rightarrow)]{\varphi,\varphi\to\psi\Rightarrow\psi}{\varphi\Rightarrow\varphi & \psi\Rightarrow\psi}}.
\]
By Corollary~\ref{cor:a-cut}, we have $\mathsf{G}^{\mathtt{a}}({\Lambda})\vdash\varphi\Rightarrow\psi$. Take a partition $\Part{\varphi:\emptyset}{\emptyset:\psi}$. Then by Lemma~\ref{lem:craig-D}, we can find an interpolant $\chi$ that satisfies the required conditions. 
\end{proof}

\section{Adding Global Modality}
\label{sec:global modality}
In this section, we introduce the framework with the global modality $\mathsf{A}$. The syntax $L^{+}$ with the global modality is defined inductively as follows:
\[
L^{+} \ni \varphi :: =  p\mid\bot\mid\neg\varphi \mid \varphi \to \psi \mid D_{G}\varphi\mid \mathsf{A}\varphi
\]
\noindent where $p \in \mathsf{Prop}$ and $G \subseteq \mathsf{Ag}$. The notion of $\mathsf{A}\varphi$ being true at $w$ in a model $M$ or a pseudo model $M$ is defined as:
\[
M,w\models \mathsf{A}\varphi\quad \text{iff}\quad \text{for all } v\in W, M,v\models\varphi.
\]
As the readers may notice, though we have defined $G\ne \emptyset$, the motivation for the global modality is to obtain the distributed knowledge of an empty group of agents, namely $D_{\emptyset}\varphi$. This is because $\bigcap_{a \in \emptyset}$ = $W \times W$. Under the semantics defined above, $D_{\emptyset}\varphi:= \A\varphi$. 

Now we introduce sequent calculi $\Ge{\mb{K45}_{D\A}}$, $\Ge{\mb{KD45}_{D\A}}$ and $\Ge{\mb{S5}_{D\A}}$, which are also expansions of $\mathbf{LK_{0}}$. For a finite multiset $\Gamma$ of formulas, we use $\A\Gamma$ to denote $\inset{\A\varphi}{\varphi\in\Gamma}$. The sequent calculus $\Ge{\mb{K45}_{D\A}}$ is $\mathbf{LK_{0}}$ expanded with the following rules:
\[
\infer[(\Rightarrow D^{\mb{K45}_{D\A}})]{D_{G_{1}}\varphi_{1},\dots,D_{G_{m}}\varphi_{m},\A\Pi\Rightarrow \A\Sigma,D_{H_{1}}\psi_{1},\dots,D_{H_{n}}\psi_{n}, D_{G}\chi}{\varphi_{1},\dots,\varphi_{m},D_{G_{1}}\varphi_{1},\dots,D_{G_{m}}\varphi_{m},\A\Pi\Rightarrow\A\Sigma,D_{H_{1}}\psi_{1},\dots,D_{H_{n}}\psi_{n},\chi},
\]
\[
\infer[(\Rightarrow\A)]{\A\Gamma\Rightarrow\A\Delta,\A\varphi}{\A\Gamma\Rightarrow\A\Delta,\varphi}
\quad
\infer[(\A \Rightarrow)]{\A \varphi,\Gamma \Rightarrow\Delta}{\varphi,\Gamma\Rightarrow\Delta}
\]
where $\bigcup^{m}_{i=1}G_{i}\cup\bigcup^{n}_{j=1}H_{j}\subseteq G$ in $(\Rightarrow D^{\mb{K45}_{D\A}})$. The sequent calculus $\Ge{\mb{KD45}_{D\A}}$ expands $\Ge{\mb{K45}_{D}}$ with the following rule:
\[
\infer[(D^{\mb{KD45}_{D\A}})]{D_{\{a\}}\Gamma, \A\Pi\Rightarrow \A\Sigma, D_{\{a\}}\Delta}{\Gamma, D_{\{a\}}\Gamma, \A\Pi\Rightarrow \A\Sigma, D_{\{a\}}\Delta}.
\]
Finally, the sequent calculus $\Ge{\mb{S5}_{D\A}}$ expands $\mathbf{LK_{0}}$ with $(\Rightarrow\A)$, $(\A\Rightarrow)$ and the following rules:
\[
\infer[(\Rightarrow D^{\mb{S5}_{D\A}})]{D_{G_1}\varphi_1, \dots, D_{G_m}\varphi_m,\A\Pi\Rightarrow \A\Sigma,D_{H_1}\psi_1, \dots, D_{H_n}\psi_n, D_G \chi}{D_{G_1}\varphi_1, \dots, D_{G_m}\varphi_m,\A\Pi\Rightarrow \A\Sigma, D_{H_1}\psi_1, \dots, D_{H_n}\psi_n,\chi}
\quad
\infer[(D\Rightarrow)]{D_{G}\varphi,\Gamma\Rightarrow\Delta}{\varphi,\Gamma\Rightarrow\Delta}.
\]
where $\bigcup^{m}_{i=1}G_{i}\cup\bigcup^{n}_{j=1}H_{j}\subseteq G$ in $(\Rightarrow D^{\mb{S5}_{D\A}})$.
The following proposition can be established similarly as Proposition~\ref{prop:soundness-D}.
\begin{proposition}
\label{prop:soundness-A}
If $\Gamma \Rightarrow \Delta$ is derivable in $\mathsf{G}(\Lambda)$ then $\bigwedge \Gamma \to \bigvee \Delta$ is valid in the corresponding class $\mathbb{P}_{\Lambda}$ of pseudo frames to $\Lambda$.
\end{proposition}

In what remains in this section, let $\Lambda\in\{\mb{K45}_{D\A},\mb{KD45}_{D\A},\mb{S5}_{D\A}\}$, and explain that the proof-theoretic results, namely the analytic cut property and Craig interpolation theorem for $\Ge{\Lambda}$ can be obtained in a similar way as in Theorems~\ref{thm:completeness-D} and~\ref{thm:craig-D}, respectively, except that a modification of the pseudo-model is required to establish the analytic cut property.

\begin{definition}
\label{dfn:RA}
Define $M^{\Lambda}_{\Xi}=(W^{\Lambda}_{\Xi}, (R^{\Lambda}_{G})_{G\in\mathsf{Grp}},V^{\Lambda}_{\Xi})$ to be the pseudo model derived from $\Xi$ in the same way as in Definition \ref{dfn:canonical pseudo model}. 
Define $\sim_{\A}$ on $W^{\Lambda}_{\Xi}$ as follows:
\[(\Gamma,\Delta)\sim_{\A}(\Pi,\Sigma)\text{ iff }\inset{\A\varphi}{\A\varphi\in\Gamma}=\inset{\A\varphi}{\A\varphi\in\Pi}.
\]
\end{definition}
It follows that $\sim_{\A}$ is reflexive, transitive, symmetric and Euclidean.
\begin{definition}
\label{dfn:filtrated canonical pseudo model}
    Let $(\Phi,\Psi)$ be a $\Xi$-partial valuation. The pseudo model 
    \[M^{\Lambda}_{\Xi(\Phi,\Psi)}=(W^{\Lambda}_{\Xi(\Phi,\Psi)}, (R^{\Lambda}_{G})_{G\in\mathsf{Grp}},V^{\Lambda}_{\Xi(\Phi,\Psi)})
    \]
    derived from $\Xi$ and $(\Gamma,\Delta)$ is defined as:
    \begin{itemize}[itemsep=0pt, parsep=0pt, topsep=5pt]
        \item $W^{\Lambda}_{\Xi(\Phi,\Psi)}=\inset{(\Gamma,\Delta)\in W^{\Lambda}_{\Xi}}{(\Phi,\Psi)\sim_{\A}(\Gamma,\Delta)}$;
        \item $R^{\Lambda}_{G}$ is defined depending on the choice of $\Lambda$:
        \begin{itemize}
            \item $\Lambda=\mb{K45}_{D\A}$: $(\Gamma,\Delta)R^{\mb{K45}_{D\A}}_{G}(\Pi,\Sigma)$ iff \\
            $\forall H\subseteq G
        (\inset{\varphi,D_{H}\varphi}{D_{H}\varphi\in\Gamma}\subseteq\Pi$ and $\inset{D_{H}\varphi}{D_{H}\varphi\in\Pi}\subseteq\Gamma$) \\
        and $\inset{\A\varphi}{\A\varphi\in\Gamma}=\inset{\A\varphi}{\A\varphi\in\Pi}$,
        \item $\Lambda=\mb{KD45}_{D\A}$: $R^{\mb{KD45}_{D\A}}_{G}:=R^{\mb{K45}_{D\A}}_{G}$,
        \item $\Lambda=\mb{S5}_{D\A}$: $(\Gamma,\Delta)R^{\mb{S5}_{D\A}}_{G}(\Pi,\Sigma)$ iff $\forall H\subseteq G(\inset{D_{H}\varphi}{D_{H}\varphi\in\Gamma}=\inset{D_{H}\varphi}{D_{H}\varphi\in\Pi})$\\
        and $\inset{\A\varphi}{\A\varphi\in\Gamma}=\inset{\A\varphi}{\A\varphi\in\Pi}$;
        \end{itemize}
        \item $(\Gamma,\Delta)\in V^{\Lambda}_{\Xi(\Phi,\Psi)}(p)$ iff $p\in\Gamma$.
    \end{itemize}
\end{definition}

\noindent The above definition of $W^{\Lambda}_{\Xi(\Phi, \Psi)}$ ensures that the global modality $\mathsf{A}$ means ``for all states'' by specifying the root $(\Phi, \Psi)$. All the necessary lemmas can be proved in a manner similar to those for distributed knowledge logic without the global modality, with the exception of the following:

\begin{lemma}
\label{lem:saturated A}
Let $(\Gamma,\Delta)$ be a $\Xi$-partial valuation in $\mathsf{G}^{\mt{a}}(\Lambda)$. 
Then $(\Gamma,\Delta)$ satisfies the following:
\begin{itemize}[itemsep=0pt, parsep=0pt, topsep=5pt]
    \item[$(\A\Rightarrow)$] $\A \varphi\in\Gamma$ implies $\varphi\in\Pi$ for all $(\Pi,\Sigma)\in W_{\Xi(\Phi,\Psi)}^{\Lambda}$;
    \item[$(\Rightarrow \A)$] $\A \varphi\in\Delta$ implies $\varphi\in\Sigma$ for some $(\Pi,\Sigma)\in W_{\Xi(\Phi,\Psi)}^{\Lambda}$;
    \item[$(D)$] $D_{G} \varphi\in\Delta$ implies $\varphi\in\Sigma$ for some $(\Pi,\Sigma)\in W_{\Xi(\Phi,\Psi)}^{\Lambda}$ such that $(\Gamma,\Delta)R^{\Lambda}_{G}(\Pi,\Sigma)$;
\end{itemize}
\end{lemma}

For item $(D)$, to find a desired witness $(\Pi, \Sigma) \in W_{\Xi(\Phi, \Psi)}^{\Lambda}$ reachable from $(\Phi, \Psi)$, adding the contexts $\mathsf{A}\Pi$ and $\mathsf{A}\Sigma$ to the inference rules $(\Rightarrow D^{\mathbf{K45}_{D}})$, 
$(D^{\mathbf{KD45}_{D}})$ and 
$(\Rightarrow D^{\mathbf{S5}_{D}})$ of $D_{G}$ plays a crucial role. By applying this lemma, we can prove the semantic completeness theorem, following the same line of argument as in the case without the global modality:

\begin{theorem}
\label{thm:completeness-A}
For any sequent $\Gamma\Rightarrow\Delta$, if $\bigwedge\Gamma\to\bigvee\Delta$ is valid in the class $\mathbb{P}_{\Lambda}$ of {all} pseudo frames which satisfy the properties corresponding to $\Lambda$, then $\mathsf{G}^{\mt{a}}(\Lambda)\vdash\Gamma\Rightarrow\Delta$.
\end{theorem}

\begin{corollary}
\label{cor:analytic-cut-property-A}
For any sequent $\Gamma\Rightarrow\Delta$, if 
$\mathsf{G}(\Lambda)\vdash\Gamma\Rightarrow\Delta$, then $\mathsf{G}^{\mt{a}}(\Lambda)\vdash\Gamma\Rightarrow\Delta$.
\end{corollary}

For the rules $(\Rightarrow \A)$ and $(\A \Rightarrow)$, we can employ the same argument in terms of Maehara's method as used for the modal logic $\mathbf{S5}$ (see, e.g.,~\cite{Ono_1998_prooftheoretic}). Moreover, 
adding the contexts $\mathsf{A}\Pi$ and $\mathsf{A}\Sigma$ to the inference rules $(\Rightarrow D^{\mathbf{K45}_{D}})$, 
$(D^{\mathbf{KD45}_{D}})$ and 
$(\Rightarrow D^{\mathbf{S5}_{D}})$ of $D_{G}$ 
does not change the outline of Maehara's argument or the information regarding the common vocabulary. Therefore, we obtain the following:

\begin{theorem}
\label{thm:craig-DA}
Let $\Lambda \in \{\mb{K45}_{D\mathsf{A}}, \mb{KD45}_{D\mathsf{A}}, \mb{S5}_{D\mathsf{A}} \}$. 
Suppose $\Ge{\Lambda}\vdash\Gamma\Rightarrow\Delta$. Then for any partition $\langle(\Gamma_{1}:\Delta_{1});(\Gamma_{2}:\Delta_{2})\rangle$ for the sequent $\Gamma\Rightarrow\Delta$, there exists a formula $\varphi$ that satisfies the following: 
$\Gamma_{1}\Rightarrow\Delta_{1},\varphi$ and $\varphi, \Gamma_{2}\Rightarrow\Delta_{2}$ are derivable in $\Ge{\Lambda}$ and $\X{\varphi}\subseteq \X{\Gamma_{1},\Delta_{1}}\cap\X{\Gamma_{2},\Delta_{2}}$ for all $\mathsf{X} \in \{\mathsf{Agt}, \mathsf{Prop}\}$. 
Therefore, the Craig interpolation theorem holds for $\Ge{\Lambda}$. 
\end{theorem}

The reader may wonder if all the sequent calculi equipped with the global modality are  semantically complete in terms of models, i.e., pseudo models $M$ = $(W,(R_{G})_{G \in \mathsf{Grp}}, V)$ such that $R_{G}$ $=$ $\bigcap_{a \in G}R_{\setof{a}}$ holds for all $G \in \mathsf{Grp}$. 
This can be done in terms of a technique called ``tree unraveling”, which transforms a pseudo-model into one that satisfies $\bigcap_{a\in G}R_{\setof{a}}=R_{G}$. 

\begin{theorem}
\label{prop:pseudomodel2model}
Let $\Lambda\in\{\mb{K45}_{D\A},\mb{KD45}_{D\A},\mb{S5}_{D\A}\}$. 
If a formula $\varphi$ is valid in the class $\mathbb{F}_{\Lambda}$ of all frames that satisfy the corresponding properties to $\Lambda$ then $\varphi$ is valid in the class $\mathbb{P}_{\Lambda}$  of all pseudo frames that satisfy the properties corresponding to $\Lambda$. 
Therefore, for any sequent $\Gamma\Rightarrow\Delta$, if $\bigwedge\Gamma\to\bigvee\Delta$ is valid in the class $\mathbb{F}_{\Lambda}$, then $\mathsf{G}^{\mt{a}}(\Lambda)\vdash\Gamma\Rightarrow\Delta$.
\end{theorem}

\begin{proof}
(Outline) The latter part follows immediately from the former part by Theorem \ref{thm:completeness-A}. In what follows, we outline the proof of the former part. The following definitions and arguments are adapted from~\cite{Fagin_1992_What,Agotnes_2021_group,Murai_2022_intuitionistic}. 
Suppose that a formula $\varphi$ is valid in the class $\mathbb{F}_{\Lambda}$. To show that $\varphi$ is valid in the pseudo-frame class $\mathbb{P}_{\Lambda}$, let us fix an arbitrary pseudo-frame $F = (W,(R_{G})_{G \in \mathsf{Grp}})$ in $\mathbb{P}_{\Lambda}$. Our goal is to establish the validity of $\varphi$ on $F$. 
Let us fix an arbitrary valuation $V$ and a state $w \in W$. 
We will show that $M,w\models \varphi$, where $M := (F,V)$. Without loss of generality, we assume that $M$ is point-generated by $w$ with respect to the binary relations $(R_{G})_{G \in \mathsf{Grp}}$. 
We define the tree unraveling $\mathtt{Tree}^{\Lambda}(M,w)$ of $M$ around $w$ as follows:\begin{itemize}\item $\mathtt{Finpath}(M,w)$ is defined as follows:
\[
\begin{aligned}
\mathtt{Finpath}(M,w) :=
\{\, &\langle w_{0},G_{1},w_{1},\dots,G_{m},w_{m}\rangle \mid {}\\
& m \ge 0,\ G_i\in\ms{Grp},\ w_{0}=w, w_{i-1}R_{G_i}w_i \text{ for all } 1\le i\le m
\,\}.
\end{aligned}
\]
We refer to an element of $\mathtt{Finpath}(M,w)$ as a ``path (from state $w$)'' and denote it by $\overrightarrow{v},\overrightarrow{u}$, etc.\item We define a binary relation $\mathcal{R}_{G}$ on $\mathtt{Finpath}(M,w)$ as follows:$\overrightarrow{w}\mathcal{R}_{G}\overrightarrow{v}$ iff  $\overrightarrow{v} = \overrightarrow{w} {}^{\frown} (H, \mathtt{tail}(\overrightarrow{w}))$ for some $H \supseteq G$, where $\mathtt{tail}(\overrightarrow{w})$ denotes the last element of $\overrightarrow{w}$, and ${}^{\frown}$ denotes concatenation. We use $\mathcal{R}_{G}^{+}$ (or $\mathcal{R}_{G}^{\ast}$) to mean the transitive closure (or the reflexive and transitive closure) of $\mathcal{R}_{G}$.
\item Depending on our choice of $\Lambda$, we define $\mathcal{R}^{\Lambda}_{G}$ as follows:
\begin{itemize}
\item For $\Lambda \in \setof{\mb{K45}_{D\A},\mb{KD45}_{D\A}}$, define $\overrightarrow{w}\mathcal{R}^{\Lambda}_{G}\overrightarrow{v}$ iff there exists $\overrightarrow{u}\in\mathtt{Finpath}(M,w)$ such that $\overrightarrow{u} \mathcal{R}_{G}^{\ast} \overrightarrow{w}$ and  $\overrightarrow{u} \mathcal{R}_{G}^{+} \overrightarrow{v}$;
\item For $\Lambda = \mb{S5}_{D\A}$, define $\overrightarrow{w}\mathcal{R}^{\Lambda}_{G}\overrightarrow{v}$ iff there exists $\overrightarrow{u}\in\mathtt{Finpath}(M,w)$ such that $\overrightarrow{u} \mathcal{R}_{G}^{\ast} \overrightarrow{w}$ and $\overrightarrow{u} \mathcal{R}_{G}^{\ast} \overrightarrow{v}$.
\end{itemize}
\item $\mathcal{V}(p)=\inset{\overrightarrow{v}\in \mathtt{Finpath}(M,w)}{\mathtt{tail}(\overrightarrow{v})\in V(p)}$ for every $p\in\ms{Prop}$.
\item Define $\mathtt{Tree}^{\Lambda}(M,w) := (\mathtt{Finpath}(M,w), (\mathcal{R}^{\Lambda}_{G})_{G \in \mathsf{Grp}}, \mathcal{V})$.
\end{itemize}
It is straightforward to verify that the frame part $(\mathtt{Finpath}(M,w), (\mathcal{R}^{\Lambda}_{G})_{G \in \mathsf{Grp}})$ of $\mathtt{Tree}^{\Lambda}(M,w)$ still belongs to $\mathbb{P}_{\Lambda}$. Moreover, 
we note that $\mathcal{R}^{\Lambda}_{G}$ $=$ $\bigcap_{a \in G} \mathcal{R}^{\Lambda}_{\setof{a}}$ holds. 
We can also establish that the mapping $\mathtt{tail}: \mathtt{Finpath}(M,w) \to W$, which sends each $\overrightarrow{v}$ to its last element $\mathtt{tail}(\overrightarrow{v})$, is a surjective bounded morphism between pseudo-models. Here, surjectivity is necessary for preserving the truth of formulas of the form $\A \psi$ (along the bounded morphism). Then we obtain $\mathtt{Tree}^{\Lambda}(M,w),\langle w \rangle  \models \varphi$ iff $M, \mathtt{tail}(\langle w\rangle) \models \varphi$, where $\mathtt{tail}(\langle w\rangle) = w$. 
Since $\mathcal{R}^{\Lambda}_{G}$ $=$ $\bigcap_{a \in G} \mathcal{R}^{\Lambda}_{\setof{a}}$ holds, 
$(\mathtt{Finpath}(M,w), (\mathcal{R}^{\Lambda}_{\setof{a}})_{a \in \mathsf{Ag}}) \in \mathbb{F}_{\Lambda}$. 
Since $\varphi$ is valid in $\mathbb{F}_{\Lambda}$, we obtain 
$(\mathtt{Finpath}(M,w), (\mathcal{R}^{\Lambda}_{\setof{a}})_{a \in \mathsf{Ag}}, \mathcal{V}), \langle w \rangle \models \varphi$, hence 
$\mathtt{Tree}^{\Lambda}(M,w),\langle w \rangle  \models \varphi$ in terms of pseudo models. Therefore, it follows from the above equivalence that $M,w\models \varphi$, as desired. 
\end{proof}

\section{Conclusion}
\label{sec:conclusion}
There are several directions for further research. First, there have been attempts to combine epistemic logic with coalition logic $\mb{CL}$, which is a logic that concerns coalitional power (see~\cite{Agotnes_2016_coalition}). It would be interesting to develop sequent calculi for $\mb{CL}$ with distributed knowledge and to investigate whether analytic cut properties can be established. Second, to the best of the authors’ knowledge, a syntactic proof of the analytic cut property in the style of~\cite{Takano_1992_subformula} is {\em not} available for our sequent calculi. The main obstacle appears to lie in the parameterization of the operator by a group $G$. However, a Gentzen-style sequent calculus for $\mb{S5}$ was introduced in~\cite{Giedra_2010_cuta}, in which the operator is not parameterized by groups. It would therefore be interesting to investigate whether Takano’s syntactic strategy can be adapted to systems combining this calculus with $\mb{CL}$. Third, Hilbert systems and cut-free sequent calculi for intuitionistic $\mb{K}$, $\mb{KT}$, $\mb{KD}$, $\mb{K4}$, $\mb{K4D}$, and $\mb{S4}$ with distributed knowledge were introduced in~\cite{Murai_2022_intuitionistic}. Building on this line of work,  public announcement expansions of intuitionistic $\mb{K}$, $\mb{KT}$, $\mb{K4}$, and $\mb{S4}$ with distributed knowledge were studied in~\cite{Murai_2024_intuitionistic}. It would therefore be worthwhile to investigate intuitionistic versions of $\mb{K45}$, $\mb{KD45}$, and $\mb{S5}$ with distributed knowledge, and to examine whether these intuitionistic systems can be expanded with public announcement logic via the approach used in~\cite{Liu_2023_nonlabelled}, as well as whether the main proof-theoretic results are preserved under such extensions. Finally, an axiomatization for comparative knowledge was introduced in~\cite{Baltag__learning}, which allows one to express that one group's (distributed) knowledge includes {\em all} of another group's (distributed) knowledge. It would therefore be interesting to develop sequent calculi for this expansion of distributed knowledge and examine whether the main proof-theoretic results established in this paper still hold.

\section*{Acknowledgements}
The authors would like to thank the three reviewers for their constructive comments and suggestions, which have helped improve the quality of the manuscript. The work of the second author is supported by the China Scholarship Council (CSC). The work of the third author was partially supported by JSPS KAKENHI Grant-in-Aid for Scientific Research (B) Grant Number JP22H00597 and Grant-in-Aid for Scientific Research (C) JP25K03537.

\bibliographystyle{eptcs}
\bibliography{reference}

\end{document}